\documentclass[11pt]{article}
\usepackage{jheppub}
\usepackage{amsmath,amsthm,amssymb,amsfonts,graphicx,mathtools}

\usepackage{hyperref}

\def\identity{{\rlap{1} \hskip 1.6pt \hbox{1}}}
\newcommand{\be}{\begin{equation}}
\newcommand{\ee}{\end{equation}}
\newcommand{\bea}{\begin{eqnarray}}
\newcommand{\eea}{\end{eqnarray}}
\newcommand{\beas}{\begin{eqnarray*}}
\newcommand{\eeas}{\end{eqnarray*}}
\newcommand{\ba}{\begin{array}}
\newcommand{\ea}{\end{array}}
\newcommand{\tr}{{\rm tr}}

\renewcommand*\d[2][]{%
	\mathrm{d}%
	\ifx\relax#1\relax\else
	\rule{-0.02em}{1.5ex}^{#1}\rule{0.08em}{0ex}\!
	\fi
	#2\,
}

\usepackage{amsthm}

\newtheorem{lemma}{Lemma}
\newtheorem{theorem}{Theorem}
\newtheorem{proposition}{Proposition}
\newtheorem{corollary}{Corollary}
\newtheorem{definition}{Definition}

\newtheorem{question}{Question}

\usepackage{quiver}

\usepackage{comment}
\usepackage{rotating}
\usepackage[final]{pdfpages}
\usepackage{tcolorbox}
\usepackage{physics}
\usepackage{bbm}
\usepackage{subfig}

\usepackage{graphicx}
  
\author[1,2]{Seraphim Jarov}
\author[1]{Mark Van Raamsdonk}
\affiliation[1]{Department of Physics and Astronomy, University of British Columbia,\\
6224 Agricultural Road, Vancouver, B.C., V6T 1Z1, Canada}
\affiliation[2]{Department of Pure Mathematics, University of Waterloo, Waterloo, Ontario, N2L 3G1, Canada}

\begin{document}

\title{Mapping the space of quantum expectation values}

\emailAdd{sjarov94@gmail.com}
\emailAdd{mav@phas.ubc.ca}
\date{October 2023}

\abstract{For a quantum system with Hilbert space ${\cal H}$ of dimension $N$ and a set $S$ of $n$ Hermitian operators ${\cal O}_i$, a basic question is to understand the set $E_S \subset \mathbb{R}^n$ of points $\vec{e}$ where $e_i = \tr(\rho {\cal O}_i)$ for an allowed state $\rho$. A related question is to determine whether a given set of expectation values $\vec{e}$ lies in $E_S$ and in this case to describe the most general state with these expectation values. In this paper, we describe various ways to characterize $E_S$, reviewing basic results that are perhaps not widely known and adding new ones. One important result (originally due to E. Wichmann) is that for a set $S$ of linearly independent traceless operators, every set of expectation values $\vec{e}$ in the interior of $E_S$ is achieved uniquely by a state of the form $\rho_{\vec{\beta}} = e^{-\sum_i \beta_i {\cal O}_i}/\tr(e^{-\sum_i \beta_i {\cal O}_i})$ for ${\cal O}_i \in S$. In fact, the map $\vec{\beta} \to \vec{E}(\vec{\beta}) = \tr(\vec{\cal O} \rho_{\vec{\beta}})$ is a diffeomorphism from $\mathbb{R}^n$ to the interior of $E_S$ with symmetric, positive Jacobian; using this fact, we provide an algorithm to invert $\vec{E}(\vec{\beta})$ and thus determine a state $\rho_{\vec{\beta}(\vec{e})}$ with specified expectation values $\vec{e}$ provided that these lie in $E_S$. The algorithm is based on defining a first order differential equation in the space of parameters $\vec{\beta}$ that is guaranteed to converge to $\vec{\beta}(\vec{e})$ in a precise way, with $|\vec{E}(\vec{\beta}(t)) - \vec{e}| = C e^{-t}$. 
}

\maketitle


\section{Introduction}

Consider a quantum mechanical system with a Hilbert space ${\cal H}$. Given a set $S$ of physical observables ${\cal O}_i$, a basic question is to understand the space of possible expectation values for these observables. 

Mathematically, we would like to understand the sets $E_S$ and $\hat{E}_S$ defined as follows.
\begin{definition}
    Given a set $S = \{ {\cal O}_1, \dots {\cal O}_n \} $ of $n$ Hermitian operators acting on ${\cal H}$ of dimension $N$, we define $E_S \subset \mathbb{R}^n$ as the union of points $(\tr(\rho {\cal O}_1),\dots,\tr(\rho {\cal O}_n))$ over all density matrices $\rho$ (unit-trace non-negative Hermitian operators) acting on ${\cal H}$. The smaller set $\hat{E}_S$ is defined by restricting to pure states with $\rho^2 = \rho$.
\end{definition}
For a two-dimensional Hilbert space, the Pauli operators $\sigma_i$ (normalized to have eigenvalues $\pm 1$) famously have expectation values confined to the unit sphere (referred to as the Bloch sphere in this context) for pure states or the unit ball for general states. 

The Bloch sphere example highlights a basic fact: $E_S$ for a set of operators is not simply the product of $E_S$ for the individual operators. The fact that specifying expectation values for certain operators constrains the expectation values for other operators is exhibited most famously in the Heisenberg Uncertainty Principle, or its generalization
\begin{equation}
\label{Uncert}
    (\langle {\cal O}_1^2 \rangle - \langle {\cal O}_1 \rangle^2) (\langle {\cal O}_2^2 \rangle - \langle {\cal O}_2 \rangle^2) \ge \langle {i \over 2}[{\cal O}_1,{\cal O}_2]\rangle^2 \; .
\end{equation}
This result can be interpreted geometrically as specifying a region of $\mathbb{R}^5$ inside of which the set $E_S$ for $S = \{  {\cal O}_1, {\cal O}_1^2, {\cal O}_2 , {\cal O}_2^2,   {i \over 2}[{\cal O}_1,{\cal O}_2]\}$ must be contained. For the Bloch sphere example, $E_S$ is precisely the intersection of sets specified by this uncertainty relation for all pairs of operators ${\cal O}_1 = n_1 \cdot \vec{\sigma}, {\cal O}_2 = n_2 \cdot \vec{\sigma}$ with $n_1 \cdot n_2 = 0$. 

In this paper, we will provide various ways to characterize $E_S$ and practical algorithms to display $E_S$ or determine whether a given set of expectation values is in $E_S$ for a given set of operators. We also provide an algorithm that explicitly constructs a state $\rho$ with a specified set of expectation values $\vec{e}$ provided that $\vec{e}$ is in the interior of $E_S$. Finally, we show that $E_S$ can always be understood as the intersection of sets specified by inequalities that generalize the uncertainty relation above.

The subject of this paper has a long history in both the physics and mathematics literature. Many of the results that we present below appear are known but perhaps not widely appreciated, so we felt it could be useful to include a broad discussion with appropriate references. However, the main results in section 6 as well as the reformulation of Theorem 2 as the more geometrical Theorem 3 are new as far as we are aware. 

\subsubsection*{Outline and summary of results}

We provide here a brief outline and summary of results. In section 2, we review a few simple examples. For the case of a two-dimensional Hilbert space, $E_S$ is always the image of a ball under an affine map. When $S$ is a set of commuting operators in a Hilbert space of general dimension, $E_S$ is a convex polytope of dimension $n$, where $n$ is the number of independent traceless operators in $\text{span}(S)$. More generally, $E_S$ can have a combination of curved sides and hyperplanar faces. As a specific example, taking ${\cal H} = \mathbb{R}^3$ and $S = \{{\cal O}_1, {\cal O}_2\}$ represented in the standard basis as 
\begin{equation}
\label{defops}
    {\cal O}_1 = \left( \begin{array}{ccc}
      -1 & 0 & 0 \cr
      0 & -1 & 0 \cr
      0 & 0 & 2
    \end{array}
    \right) \qquad \qquad {\cal O}_2 = \left( \begin{array}{ccc}
      0 & 1 & 1 \cr
      1 & 0 & 0 \cr
      1 & 0 & 0
    \end{array}
    \right) \; ,
\end{equation}
the set $E_S$ is the region of $\mathbb{R}^2$ shown in Figure \ref{fig:example}.
\begin{figure}
    \centering
    \includegraphics[width = .5 \textwidth]{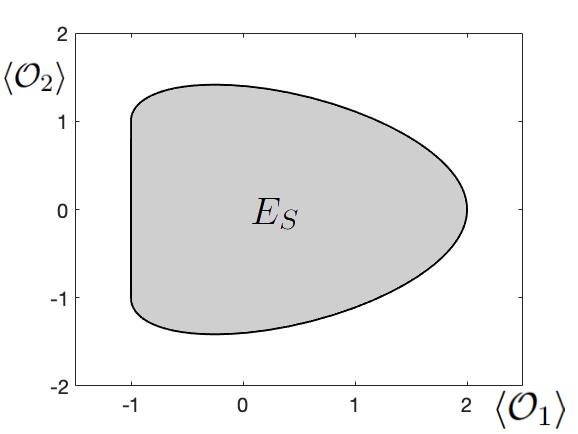}
    \caption{The set $E_S$ for operators (\ref{defops}) acting on ${\cal H} = \mathbb{R}^3$.}
    \label{fig:example}
\end{figure}
In section 3, we show that $E_S$ is always a compact, convex subset of $\mathbb{R}^n$ that is the convex hull of $\hat{E}_S$. Also, $E_S$ is the image under an affine transformation of the set $E_{\tilde{T}}$, where $\tilde{T}$ is an orthogonal set of basis elements for the full space of traceless Hermitian operators acting on ${\cal H}$. The $N^2-1$ dimensional set $E_{\tilde{T}}$ is a higher-dimensional generalization of the Bloch ball, and is the convex hull of the $(2N-2)$-dimensional $\hat{E}_{\tilde{T}}$ which generalizes the Bloch sphere. For this latter set, we provide an algebraic characterization as the solution of a system of polynomial equations.

\subsubsection*{The set $E_S$ as an intersection of half-spaces}

An obvious fact about $E_S$ is that for any linear combination ${\cal O}_{\hat{e}} = \sum_i \hat{e}_i \cdot {\cal O}_i$ of operators in $S$ and any state $\rho$, the expectation value of ${\cal O}_{\hat{e}}$ for $\rho$ is larger than or equal to the minimum eigenvalue of ${\cal O}_{\hat{e}}$. This implies $E_S$ is contained in the the half space $H_{\hat{e}}$ defined by  $\hat{e} \cdot \vec{x} \ge \lambda_{min}({\cal O}_{\hat{e}})$. In section 4, we prove that $E_S$ is exactly the intersection of such half-spaces over all possible unit vectors $\hat{e} \in \mathbb{R}^n$. In general, the dimension of $E_S$ is equal to the number of linearly independent traceless operators in $\text{span}(S)$. Restricting $S$ to a set of linearly independent traceless operators, $E_S$ is a dimension $n$ subset of $\mathbb{R}^n$ and has a boundary that we also characterize. The main result (mostly described previously in \cite{wichmann1963density,Szymanski:2022sgn}) is:
\begin{theorem}
\label{boundary}
    For a set $S = \{{\cal O}_i\}$ of $n$ linearly independent traceless Hermitian operators and $\hat{e} \in \mathbb{R}^n$, define ${\cal O}_{\hat{e}} = \sum_i \hat{e}_i \cdot {\cal O}_i$. Define $\lambda_{min}({\cal O}_{\hat{e}})$ to be the minimum eigenvalue for ${\cal O}_{\hat{e}}$ and ${\cal H}_{\hat{e}}$ the space of eigenvectors with this eigenvalue. 
    \begin{enumerate}
    \item
    The set $E_S$ can be described as an intersection of half-spaces
    \begin{equation}
    \label{first}
    E_S = \bigcap_{\substack{\hat{e} \in \mathbb{R}^n \\ |\hat{e}| = 1}} H_{\hat{e}} \quad\text{with} \quad H_{\hat{e}} = \{\vec{x} | \hat{e} \cdot \vec{x} \ge \lambda_{min}({\cal O}_{\hat{e}}) \}.
    \end{equation}
    \item The boundary of $E_S$ is 
    \begin{equation}
    \label{ESset1}
    \partial E_S = \bigcup_{\substack{\hat{e} \in \mathbb{R}^n \\ |\hat{e}| = 1}}  B_{\hat{e}} 
    \end{equation}  
    where $B_{\hat{e}} = \partial H_{\hat{e}} \cap E_S$ is a compact convex subset of the hyperplane $\hat{e} \cdot \vec{x} = \lambda_{min}({\cal O}_{\hat{e}})$. 
    \item A state $\rho$ has expectation values in $B_{\hat{e}}$ if and only if $\rho$ is in $G_{\hat{e}}$, the set of states with minimum expectation value for ${\cal O}_{\hat{e}}$. 
    \item The set $B_{\hat{e}}$ is equal to $E_{S_{\hat{e}}}$ for the set of operators 
     \begin{equation}
          S_{\hat{e}} = \{\pi_{\hat{e}} {\cal O}_i \pi_{\hat{e}} \} 
     \end{equation}
     acting on ${\cal H}_{\hat{e}}$, where $\pi_{\hat{e}}$ is the projection from ${\cal H}$ to ${\cal H}_{\hat{e}}$.
     \end{enumerate}
\end{theorem}
For the example above $B_{\hat{e}}$ is a line segment for $\hat{e} = \hat{e}_1 \equiv (1,0)$,   and a single point for all other directions $\hat{e}$, as shown in Figure \ref{fig:supporting}.
\begin{figure}
    \centering
    \includegraphics[width = .9 \textwidth]{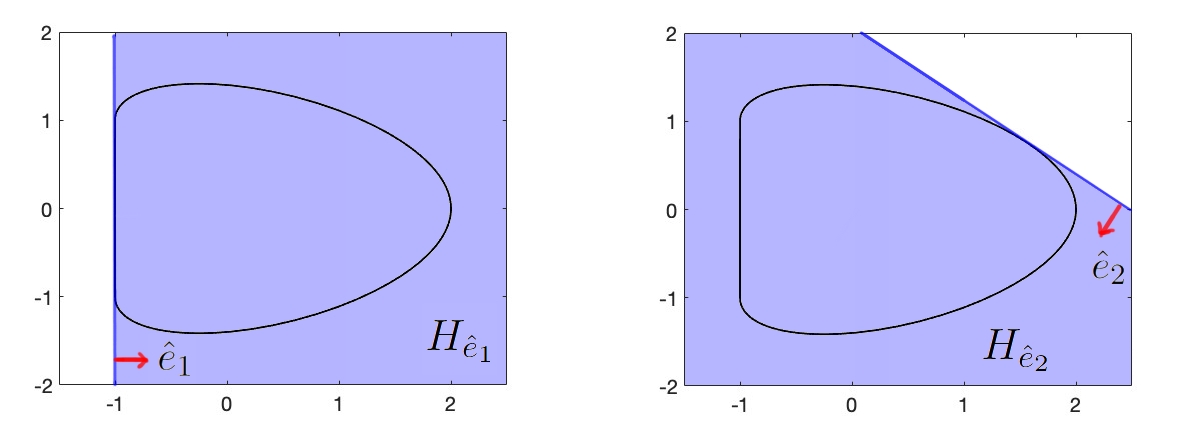}
    \caption{The set $E_S$ is the intersection of half spaces $H_{\hat{e}}$ defined by $\hat{e} \cdot \vec{x} \ge \lambda_{min}(\hat{e} \cdot {\cal O})$ over all unit vectors $\hat{e}$. The intersection of the hyperplane $\partial H_{\hat{e}}$ with $E_S$ is a single point when $\hat{e} \cdot {\cal O}$ has a non-degenerate minimum eigenvalue.}
    \label{fig:supporting}
\end{figure}

\subsubsection*{The set $E_S$ as the image of a diffeomorphism from $\mathbb{R}^n$.}

Result 3 in Theorem \ref{boundary} shows that the boundary points in $E_S$ correspond to ground states of Hamiltonians of the form ${\cal O}_{\hat{e}}$ for ${\cal O}_i$ in $S$. 
In section 5, we show that points in the interior of $E_S$ excluding the origin (again, for a set of linearly independent traceless operators) are in one-to-one correspondence with $T > 0$ thermal states for Hamiltonians of the form ${\cal O}_{\hat{e}}$, i.e. states of the form
\begin{equation}
\label{betastates}
    \rho_{\vec{\beta}} = {e^{- \beta {\cal O}_{\hat{e}} }\over \tr e^{- \beta {\cal O}_{\hat{e}}} } \equiv {e^{- \sum_i \beta_i {\cal O}_i} \over \tr(e^{- \sum_i \beta_i {\cal O}_i})} \; 
\end{equation}
where $\beta = 1/T$ and $\vec{\beta} = \beta \hat{e}$. With this definition, we have a bijection between points in the interior of $E_S$ and parameters $\vec{\beta} \in \mathbb{R}^n$. The states $\rho_{\vec{\beta}}$ are the states of maximum entropy that achieve the desired expectation values. The bijection between the parameters $\beta_i$ and the expectation values $E_i$ also exhibits some nice properties. These are summarized by the following theorem (originally due to Wichmann \cite{wichmann1963density}):
\begin{theorem}
\label{expthm}
If $S = \{{\cal O}_i \}$  is a set of $n$ linearly independent operators with $\identity \notin \text{span}(S)$, the map $\vec{\beta} \to \vec{E}(\vec{\beta}) = \tr(\rho_{\vec{\beta}} \vec{ \cal O}  )$ 
is a diffeomorphism between $\mathbb{R}^n$ and the interior of $E_S$ with symmetric negative definite Jacobian $\partial E_i / \partial \beta_j$ .
\end{theorem}
For the example of Figure \ref{fig:example}, this map is depicted in Figure \ref{fig:expmap}. 

\begin{figure}
    \centering
    \includegraphics[width = \textwidth]{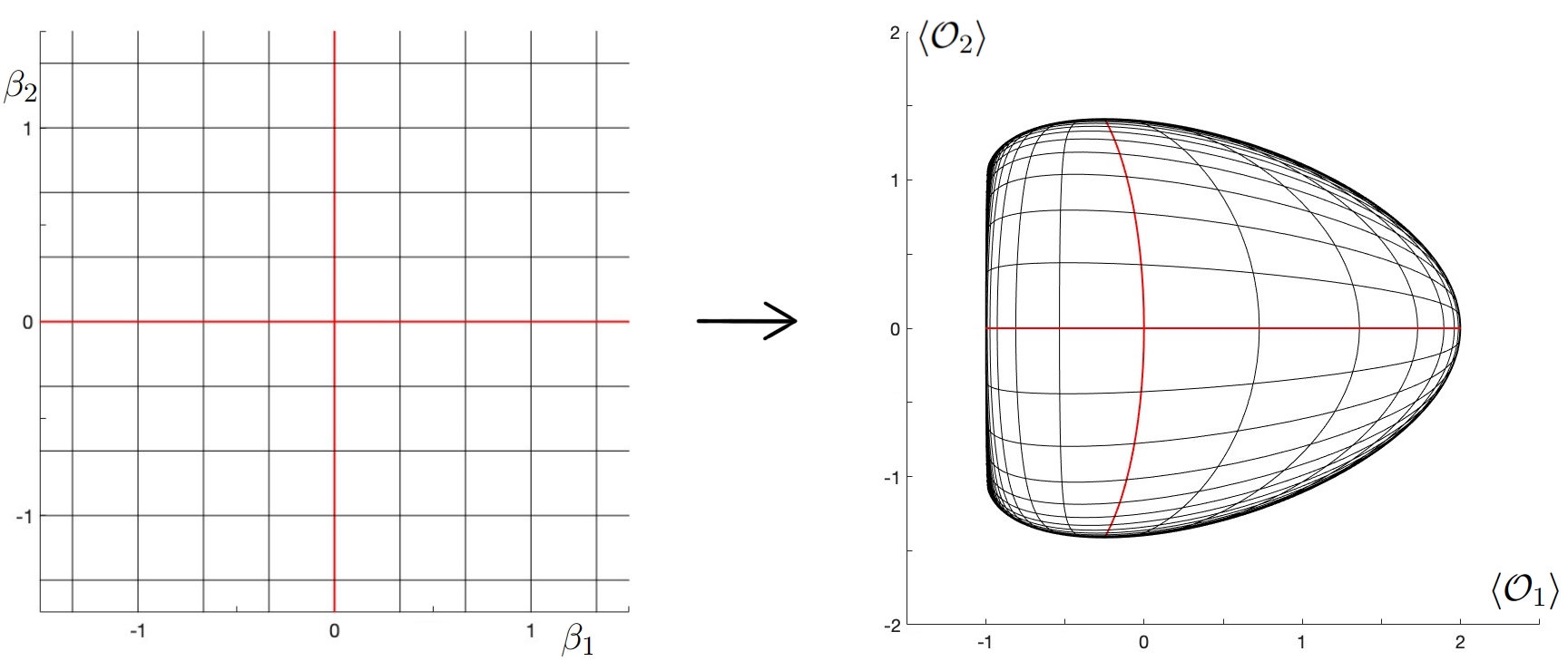}
    \caption{The map $\vec{\beta} \to \vec{E}(\vec{\beta}) = \tr(\rho_{\vec{\beta}} \vec{ \cal O})$ for $\rho_{\vec{\beta}}$ defined in (\ref{betastates})
is a diffeomorphism between $\mathbb{R}^n$ and the interior of $E_S$ with symmetric negative definite Jacobian $\partial E_i / \partial \beta_j$ .}
    \label{fig:expmap}
\end{figure}

\subsubsection*{Geometrical description of Theorem 2}

As we explain in section 5, Theorem \ref{expthm} is equivalent to a more geometrical underlying result. To understand this, we recall that quantum states can naturally be understood as elements of the dual vector space $\mathfrak{h}^*$ to the real vector space $\mathfrak{h}$ of Hermitian operators. Elements of $\mathfrak{h}^*$ are linear maps from $\mathfrak{h}$ to $\mathbb{R}$; quantum states correspond to the subset ${\cal D}^*$ of $\mathfrak{h}^*$ defined by maps ${\cal O} \to \tr(\rho {\cal O})$ for non-negative unit trace Hermitian operators $\rho$.

Given an element of $\mathfrak{h}^*$, we can restrict its action to a subspace $\mathfrak{h}_S \subset \mathfrak{h}$ to define an element of $\mathfrak{h}_S^*$. Starting from the full set of quantum states ${\cal D}^* \subset \mathfrak{h}^*$ we obtain via this restriction a subset ${\cal D}^*_S$ of $\mathfrak{h}^*_S$. The set ${\cal D}^*_S$ tells us all the possible ways to assign expectation values to operators in $\mathfrak{h}_S$ that arise from considering a valid quantum state. Thus, our original question of characterizing $E_S$ for a set $S$ of operators is equivalent to the question of characterizing the set ${\cal D}^*_S \subset \mathfrak{h}_S^*$ for a subspace $\mathfrak{h}_S = \text{span}(S)$.

If $\mathfrak{h}_S$ is a subspace of $\mathfrak{h}$ that does not contain the identity operator, the set ${\cal D}^*_S$ is a full dimension subset of $\mathfrak{h}^*$. Theorem 2 is equivalent to the statement that the interior of ${\cal D}^*_S$ is the image of a certain map from $\mathfrak{h}_S$ to $\mathfrak{h}_S^*$ with nice properties. 
To define this, consider the function $f: \mathfrak{h} \to \mathbb{R}$ defined by
\begin{equation}
f({\cal O}) = \log \tr e^{\cal O}.
\end{equation}
The derivative $df$ is a map from $\mathfrak{h} \to \mathfrak{h}^*$ defined by
\begin{equation}
df({\cal O}) = {\tr( \cdot \; e^{\cal O}) \over \tr(e^{\cal O})} \; .
\end{equation}
Since $e^{\cal O}/\tr(e^{\cal O})$ is positive and has unit trace, each point in the image of $df$ is a valid quantum state. Thus, the image of $df$ lies in ${\cal D}^*$. Starting from $df$, we can define a map $df_S: \mathfrak{h}_S \to \mathfrak{h}_S^*$ by restricting $df$ to act on $\mathfrak{h}_S$ and restricting the maps in the image to act on $\mathfrak{h}_S$.\footnote{More formally, we can define $df_S = \iota^*_S \circ df \circ \iota_S$ where  $\iota_S : \mathfrak{h}_S \to \mathfrak{h}$ is  the inclusion map and $\iota^*_S$ is the restriction map defined by $\iota^*_S \circ w = w \circ \iota_S$ for $w \in \mathfrak{h}^*$.} Then we have
\begin{theorem}
\label{dualprop}
    For any subspace $\mathfrak{h}_S \subset \mathfrak{h}$ with $\identity \notin \mathfrak{h}_S$, the map $ df_S: \mathfrak{h}_S \to \mathfrak{h}_S^*$ is a diffeomorphism from $\mathfrak{h}_S$ to the interior of $\mathfrak{D}^*_S$ with symmetric, positive definite Jacobian.
\end{theorem}
This is the geometrical version of Theorem \ref{expthm}, expressed without reference to a particular basis. Taking $\mathfrak{h}_S$ to be the subspace of traceless Hermitian operators, or any $N^2-1$ dimensional subspace not containing the identity, the theorem implies that the image is the interior of the full space of density operators.

\subsubsection*{The inverse problem}

In section 6, we describe how the results of section 5 may be used to determine whether or not a given set of expectation values $\vec{e}$ is in $E_S$ and if so, to describe the most general state with these expectation values. The Quantum Marginal Problem is a special case of this, as we review below.

A basic observation is that if $C(\vec{E})$ is any function (e.g. $C(\vec{E}) = |\vec{E} - \vec{e}|^2$)  with global minimum 0 at $\vec{e}$ that is its only critical point, the function $C(\vec{E}(\vec{\beta}))$ has a critical point in $\mathbb{R}^n$ if and only if $\vec{e}$ is in the interior of $E_S$, and in this case the critical point is the global minimum. This relies on the fact that $\vec{E}(\vec{\beta})$ is a diffeomorphism so that $\partial_i E_j$ is non-singular. We can now make use of any algorithm (e.g. gradient descent) that will succeed in minimizing $C(\vec{E}(\vec{\beta}))$ in cases where there is a global minimum that is the only critical point. If a minimum is found at some $\vec{\beta}$, the state $\rho_{\vec{\beta}}$ achieves the desired expectation values. This approach was discussed and used (without proof of its effectiveness) recently in \cite{enhanced_neg_energ1,enhanced_neg_energ} and \cite{zeng2023maximum}. 

\begin{figure}
    \centering
    \includegraphics[width = \textwidth]{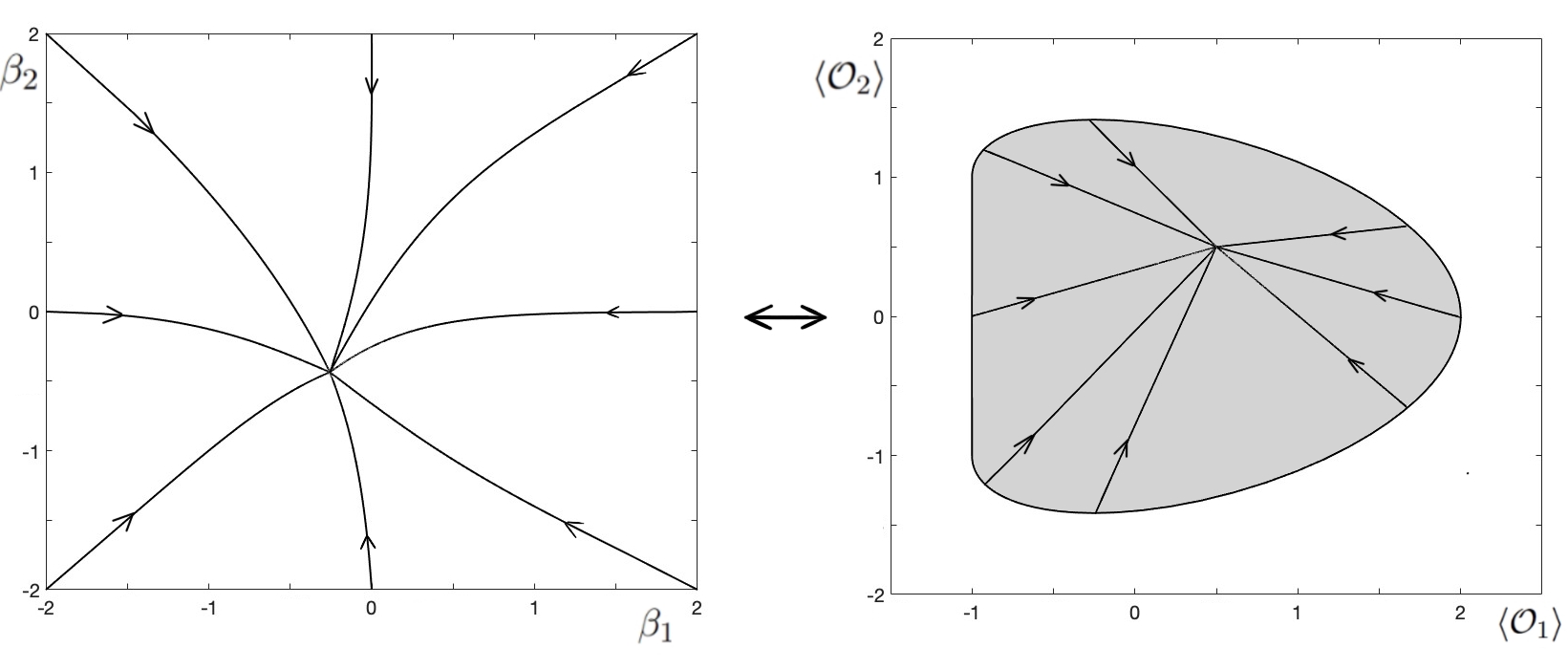}
    \caption{The flow defined in (\ref{flow1}), corresponding to gradient flow for the function $\Delta(\vec{E}) = |\vec{E}-\vec{e}|^2/2$ in the space of expectation values, we have chosen $\vec{e} = (0.5,0.5)$. Pictured are the flow lines passing through $\vec{\beta} = (\pm 2, 2), (\pm 2, 0), (\pm 2,-2), (0, \pm 2) $.}
    \label{fig:flowplot}
\end{figure}

A closely related approach is to define a flow in the space of states (\ref{betastates}) that is the image under the inverse map $\vec{E}(\vec{\beta}) \to \vec{\beta}$ of the negative gradient flow for the function $\Delta(\vec{E}) = |\vec{E}-\vec{e}|^2/2$ in the space of expectation values. When $\vec{e}$ is an allowed set of expectation values in the interior of $E_S$, this flow for an arbitrary initial $\vec{\beta}$ converges exponentially fast to the unique $\beta$ for which the state $\rho_{\vec{\beta}}$. The main result is:
\begin{theorem}
Consider the flow resulting from the equation 
\begin{equation}
\label{flow1}
\partial_t \beta_i = -(J^{-1})_{ji} (E_j(\vec{\beta}) - e_j) \qquad \qquad J_{ij} = \partial_i E_j \; .
\end{equation}
with some arbitrary initial $\vec{\beta}$. Then along the flow, we have $\Delta(\vec{E}(\vec{\beta}(t))) = \Delta_0 e^{-2t}$ and
\begin{enumerate}
    \item
    If $\vec{e} \in int(E_S)$,
the flow converges to the unique $\vec{\beta} = \vec{\beta}_\infty$ with $\vec{E}(\vec{\beta}_\infty) = \vec{e}$. 
\item
If $\vec{e} \in \partial E_S$, $|\vec{\beta}(t)|$ diverges for $t \to \infty$ and $\lim_{t \to \infty} \Delta(\vec{E}(\vec{\beta}(t))) = 0$.
\item
If $\vec{e} \notin \partial E_S$, $|\vec{\beta}(t)|$ diverges at some finite time $t_0$ and $\lim_{t \to t_0} \Delta(\vec{E}(\vec{\beta}(t))) > 0$.
\end{enumerate}
\end{theorem} 
Using numerical algorithms to implement the flow in (\ref{flow1}) provides a practical method to find a state $\rho_{\vec{\beta}}$ achieving a give set of expectation values or showing that such a state does not exist. Figure \ref{fig:flowplot} shows an example of the flow for the operators $\ref{defops}$ with the choice $\vec{e} = (0.5,0.5)$. For a pair of random Hermitian operators on a 1000 dimensional Hilbert space, we find that a discretized version of the flow (\ref{flow1}) with a relatively large time step $\delta t \sim 0.4$ typically converges to $\Delta < 10^{-3}$ in $\sim 10$ time steps. In cases where the expectation values can be realized, Proposition \ref{generalstate} in section 6 explains how to describe the most general state with these expectation values starting from the state $\rho_{\vec{\beta}}$.

\subsubsection*{Uncertainty relations and positivity constraints}

In section 7, we provide some insight into how uncertainty relations constrain the allowed expectation values of operators. We recall how a very general class of uncertainty relations can be derived from the basic result that $\langle ({\cal O} - \langle {\cal O} \rangle)^2 \rangle$ is non-negative for any operator and any state. This positivity implies the positivity of the matrix
\begin{equation}
    M_{ij} = \tr(\rho {\cal O}_i {\cal O}_j) - \tr(\rho {\cal O}_i) \tr(\rho {\cal O}_j)
\end{equation}
which can be expressed as a matrix $M_{ij}(\vec{x})$ of quadratic polynomials in $x^i = \langle {\cal O} \rangle$ via the operator algebra. In theorem (\ref{uncert}), we show (Theorem 7) that for $\tilde{T}$ some basis of traceless operators, $\vec{x} \in E_{\tilde{T}}$ if and only if the matrix $M_{ij}(\vec{x})$ is non-negative. 

\subsubsection*{Related works}

There have been many works in the past relating to constraining the space of possible expectation values for a set of operators. In the physics literature, the problem that we consider here appears to have been discussed first in \cite{wichmann1963density}. In fact, many of the results we review here in sections 2 through 5 were already derived in that paper.

In operator theory, the set $E_S$ (or sometimes the set $\hat{E}_S$) is referred to as the ``joint numerical range'' of the operators in $S$. A recent work with extensive references to the mathematics literature is \cite{Plaumann_2021}. A recent overview of the physics literature with many references is the thesis \cite{Szymanski:2022sgn}. Detailed results characterizing $E_S$ for the special case of three and four dimensional Hilbert spaces, with some more general discussion may be found in \cite{Szyma_ski_2018, eltschka2021shape,bengtsson2012geometry,goyal2016geometry,sharma2021four, boya2008geometry}. States of the form (\ref{betastates}) have found many applications in past work, e.g. \cite{Swingle_2014, yunger2016microcanonical,enhanced_neg_energ,zeng2023maximum}. They are referred to as generalized Gibbs states or non-abelian thermal states \cite{yunger2016microcanonical}. A mathematical discussion of properties of these families of states may be found in \cite{Weis}. 

\section{Simple examples}

To develop intuition, we begin by discussing a few special cases. 

\subsection{Single operator}

For a single operator ${\cal O}$, $E_{\cal O}$ is the interval $[\lambda_{min}, \lambda_{max}]$ bounded by the minimum and maximum eigenvalues of ${\cal O}$. We can attain all of these expectation values by considering the one parameter family of states $\rho = p v_{min} v_{min}^\dagger + (1-p) v_{max} v_{max}^\dagger$ for $p \in [0,1]$, where $v_{min}$ and $v_{max}$ are eigenvectors corresponding to the minimum and maximum eigenvalues.  

For general $S = \{{\cal O}_i\}$ the expectation value of each operator is constrained by this result, so we have that $E_S \subset E_{{\cal O}_1} \times \cdots \times E_{{\cal O}_n}$.

\subsection{Two-dimensional Hilbert space}

For a two-dimensional Hilbert space, we can represent a general density operator explicitly in terms of the Pauli operators $\sigma_i$ (normalized to have eigenvalues $\pm 1$) as 
\begin{align*}
    \rho = {1 \over 2} (\identity + x_i \sigma_i) =  {1 \over 2}\begin{pmatrix}
        1+ x_3 & x_1-i x_2\\
        x_1 + i x_2 & 1 - x_3
    \end{pmatrix}\quad\text{with}\quad \vec{x} \in \mathbb{R}^3, |\vec{x}| \le 1.
\end{align*}
Here, we are using the standard matrix representation associated with the orthonormal basis of $\sigma_3$ eigenvectors, and the latter constraint comes from requiring that the eigenvalues of $\rho$ are non-negative.

Taking $S = \{ \sigma_1, \sigma_2, \sigma_3 \}$, we can use the fact that $\sigma_i$ are traceless and satisfy $\tr(\sigma_i \sigma_j) = 2 \delta_{ij}$, to obtain 
\begin{equation}
    (\langle \sigma_1 \rangle,\langle \sigma_2 \rangle,\langle \sigma_3 \rangle ) = (\tr(\rho \sigma_1),\tr(\rho \sigma_2),\tr(\rho \sigma_3)) = (x_1,x_2,x_3)
\end{equation} 
so we conclude that
\begin{equation}
    E_S = \{\vec{x}\subset \mathbb{R}^3|\vec{x} \le 1\} \; .
\end{equation} 
The pure states (with $\rho^2 = \rho$) correspond to $|\vec{x}|=1$, so $\hat{E}_S$ is the boundary of the ball. 

As another example, consider an arbitrary pair of Hermitian operators ${\cal O}_1 = c^{(1)}_0 \identity + c^{(1)}_i \sigma_i$ and ${\cal O}_2 = c^{(2)}_0 \identity + c^{(2)}_i \sigma_i$ where $c_i$ are real. For the general density matrix above, we have
\begin{equation}
(\langle {\cal O}_1 \rangle,\langle {\cal O}_1 \rangle) = (c_0^{(1)},c_0^{(2)}) + (x_i c^{(1)}_i,x_i c^{(2)}_i).
\end{equation}
This is the general affine map from $\mathbb{R}^3$ parameterized by $x_i$ to $\mathbb{R}^2$. Under such a map, image of the unit sphere and the unit ball are the same, and this image is always an ellipse plus its interior, or its degeneration to a line segment or point. These are then the most general possibilities for the geometry of $E_S$ and $\hat{E}_S$ when $S$ is a pair of Hermitian operators acting on a two-dimensional Hilbert space.

\subsection{Commuting operators}
\label{sec:polygon}

As another simple example, we consider a general set of mutually commuting Hermitian operators acting on a Hilbert space of general dimension $N$. We have
\begin{proposition}
\label{polyhedron}
    For $S$ a set of commuting operators $\{{\cal O}_1,...,{\cal O}_n\}$ acting on an $N$-dimensional Hilbert space, $E_S$ is the convex hull of the points $\vec{x}^{(i)} =  (\lambda^{(i)}_1, \dots, \lambda^{(i)}_n)$ where $\vec{\lambda}^{(i)}$ is the vector of $\vec{O}$ eigenvalues for the $i$th vector in a set of orthonormal simultaneous eigenvectors for the operators in $S$.
\end{proposition}
\begin{proof}
    Since the operators in $S$ commute, we can choose an orthonormal basis of vectors $v^{(i)}$ in ${\cal H}$ that are each eigenvectors for all operators in  $S$. Thus, we have ${\cal O}_k v^{(i)} = \lambda^{(i)}_k v^{(i)}$ for some $\lambda^{(i)}_k$s. We can represent a general density operator $\rho = \sum_{ij} \rho_{ij} v^{(i)} (v^{(i)})^\dagger$ where $\rho_{ij}$ is a non-negative unit trace matrix. For this density operator, we have that
    \begin{equation}
    \tr(\rho {\cal O}_k) = \sum_i \rho_{ii} \lambda^{(i)}_k \; .
    \end{equation}
    Since $\tr(\rho) = 1$, we must have  $\sum_i \rho_{ii} = 1$ and since $\rho$ is positive, we have $\rho_{ii} = (v^{(i)})^{\dagger} \rho v^{(i)} \ge 0$. Thus $\langle \vec{\cal O} \rangle$ is contained in the convex hull of the points $\vec{\lambda}^{(i)}$.

    Conversely, consider any point $\vec{x}$ in the convex hull of the points $\vec{\lambda}^{(i)}$. We can represent $\vec{x}$ as $\sum_i p_i \vec{\lambda}^{(i)}$ for non-negative $p_i$ with $\sum p_i = 1$. Choosing the state with $\rho_{ij}= p_i \delta_{ij}$, we see that $\langle \vec{O} \rangle = \vec{x}$. Thus any point in the convex hull of the points $\vec{\lambda}^{(i)}$ is in $E_S$. 
\end{proof}
For the case of two operators, the region $E_S$ is thus a convex polygon in $R^2$ where the number of vertices is less than or equal to the dimension $N$ of the Hilbert space. 

\section{Basic properties of $E_S$}

We now consider the general case. We would like to characterize $E_S$ for a set $S$ of $n$ Hermitian operators ${\cal O}_i$ acting on a Hilbert ${\cal H}$ space of dimension $N$. 

Let $\mathfrak{h}$ be the $N^2$-dimensional real vector space of Hermitian operators on ${\cal H}$ and $\mathfrak{h}_0$ be the subspace of traceless operators. Quantum states can be associated with particular operators in $\mathfrak{h}$: a general state corresponds to a unit-trace non-negative hermitian operator $\rho$ (the density operator for that state), while pure states correspond to states with $\rho^2 = \rho$ (or equivalently, states of the form $\rho = v v^\dagger$ for some vector $v \in {\cal H}$ with $v^\dagger v =1$). 

The full space of states has dimension $N^2-1$, while the space of pure states has dimension $2N-2$.\footnote{To see the latter, we recall that the $N$-dimensional complex vector $v$ is normalized and that the transformation $v \to e^{i \theta} v$ does not affect the state.} It will be useful to note the following:
\begin{proposition}
\label{convex}
    As a subset of $\mathfrak{h}$, the density operators form the convex hull of the set of pure states. 
\end{proposition}
\begin{proof}
Any density operator has an orthonormal set of eigenvectors $v_i$ with non-negative eigenvalues $p_i$ summing to one. In terms of these, the density operator can be represented as $\sum_i p_i v_i v_i^\dagger$ and thus lies in the convex hull of the pure states. Conversely, for any set of pure states $\rho_i = v_i v_i^\dagger$ and any set of non-negative $p_i$ summing to one, the state $\rho = \sum_i p_i v_i v_i^\dagger$ is a valid density operator, since $\tr(\rho) = \sum_i p_i v_i^\dagger v_i  = 1$ and for any $w \in {\cal H}$, $w^\dagger \rho w = \sum_i p_i |w^\dagger v_i |^2 \ge 0$.
\end{proof}

It will be convenient to choose a basis $T$ for $\mathfrak{h}$. We will consider $T_0 = \identity$ together with $N^2-1$ orthogonal traceless Hermitian matrices (collectively denoted as $\tilde{T}$), normalized so that 
\begin{equation}
\label{tnorm}
    \tr(T_a T_b) = N \delta_{ab} \; .
\end{equation}
A general Hermitian operator can be represented as
\begin{equation}
    {\cal O} = c_0 \identity + \sum_a c_a \tilde{T}_a
\end{equation}
while a general state can be represented as
\begin{equation}
\label{defrho}
    \rho = {1 \over N} (\identity + x_a \tilde{T}_a) \; .
\end{equation}
The coordinates $x_a = \tr(\rho \tilde{T}_a)$ have restrictions arising from the positivity of $\rho$. The allowed set is $E_{\tilde{T}} \subset \mathbb{R}^{N^2 - 1}$. We now have a useful general result: 
\begin{proposition}
    For a set $S = \{{\cal O}^{(i)} \}$ of $n$ Hermitian operators, the set $E_S \subset \mathbb{R}^n$ is the image of under some affine transformation of $E_{\tilde{T}}$ where $\tilde{T}$ is any set of orthogonal basis elements for $\mathfrak{h}_0$, normalized as in (\ref{tnorm}). Also, $\hat{E}_S$ is the image of $\hat{E}_{\tilde{T}}$. 
\end{proposition}
\begin{proof}
    Representing 
    \begin{equation}
         {\cal O}^{(i)} = c_0^{(i)} \identity + \sum_a c_a^{(i)} \tilde{T}_a 
    \end{equation}
    and a general state $\rho$ as in (\ref{defrho}) we have
    \begin{equation}
    \tr(\rho {\cal O}^{(i)}) = c^{(i)}_0 + x_a c^{(i)}_a.
    \end{equation}
    The coordinates $x_a$ for allowed states $\rho$ define the set $E_{\tilde{T}}$, so we see that $E_S$ is the image of this under the affine transformation defined by $c_a^{(i)}$. $\hat{E}_S$ is defined to be the subset of $E_S$ obtained by restricting to pure states; this corresponds to restricting the domain of the affine map to $\hat{E}_{\tilde{T}}$.
\end{proof}
We further note that collections of $n$ operators are in one-to-one correspondence with possible affine maps, so the possibilities for $E_S$ are exactly the possible images of $E_{\tilde{T}}$ under an affine map. It follows immediately from Proposition \ref{convex} and the previous result that
\begin{corollary}
The set $E_S$ is the convex hull of the set $\hat{E}_S$.
\end{corollary}
These results generalize our observations in the two-dimensional case that possible sets $\hat{E}_S$ ($E_S$)  are images of the Bloch sphere (ball) under affine transformations. In that case, the Bloch ball ($E_{\tilde{T}}$ for the Pauli operators) was simply the interior of the Bloch sphere $\hat{E}_{\tilde{T}}$. 

In higher dimensions, $E_{\tilde{T}}$ of dimension $N^2 - 1$ is the convex hull of the $(2N-2)$-dimensional $\hat{E}_{\tilde{T}}$ rather than the interior. 

In two-dimensions, $\hat{E}_{\tilde{T}}$ has a simple algebraic characterization as the unit sphere $\sum x_i^2 =1$ when choosing $\tilde{T}$ as the Pauli operators. In higher dimensions, we can still give an algebraic characterization of $\tilde{T}$. Since $\tr(\rho \tilde{T}_a) = x_a$, the desired characterization is a set of equations on $x_a$ equivalent to the condition that $\rho$ in (\ref{defrho}) represents a pure quantum state. We can convert this condition into a set of algebraic equations by inserting (\ref{defrho}) into one of the following equivalent conditions:
\begin{itemize}
\item
The condition that $\rho^2 = \rho$. This results in a set of quadratic equations in the coordinates $x_a$.
\item 
The condition that the characteristic polynomial for $\rho$ is $\lambda^{N-1}(\lambda -1)$. This results in a set of $N-1$ equations of degree $2,3, \dots, N$ in the coordinates $x_a$.
\item 
The condition that all $2 \times 2$ subdeterminants of $\rho$ vanish in the matrix representation associated with some basis of $\mathfrak{h}$. This results in a set of quadratic equations in the coordinates $x_a$. 
\end{itemize}

We provide a more explicit derivation of the resulting equations in Appendix A. A more comprehensive discussion of the algebraic characterization of $\hat{E}_{\tilde{T}}$ can be found in \cite{Plaumann_2021} and references therein.

Geometrically, $\hat{E}_{\tilde{T}}$ can be understood as the intersection of surfaces specified by the various equations. The resulting geometry is evidently quite complicated. According to Proposition \ref{polyhedron}, we can obtain an arbitrary polyhedron with $N$ vertices in $\mathbb{R}^n$ ($n \le N$) as the image of $\hat{E}_{\tilde{T}}$ under an affine map from $\mathbb{R}^N$ to $\mathbb{R}^n$. 

In summary, we have seen that $E_S$ is the image of an affine transformation from $E_{\tilde{T}}$, a higher-dimensional generalization of the Bloch ball. Further, $E_{\tilde{T}}$ is the convex hull of $\hat{E}_{\tilde{T}}$ which generalizes the Bloch sphere and can be specified by a set of algebraic equations. This characterization of $E_S$ is not particularly convenient from a calculational perspective, since it is difficult to solve the non-linear equations that determine $\hat{E}_{\tilde{T}}$. 

In the following sections, we provide a  characterization of $E_S$ that is more useful for calculations when the dimension of ${\cal H}$ is larger than two. For simplicity, we will restrict to the case where $S$ is a set of linearly independent traceless operators. For more general $S$, we can always express the operators as ${\cal O}_i = c^{(0)}_i \identity + c^{(a)}_i \hat{\cal O}_a$ where $\hat{\cal O}_a$ are traceless and linearly independent. Then $E_S$ for the more general case is the image of an affine map on $E_S$ for these independent traceless operators. With this restriction, we have
\begin{proposition}
\label{closedconvex}
    If $S$ is a set of $n$ linearly independent traceless Hermitian operators, $E_S$ is a compact, convex subset of $R^n$ of dimension $n$ containing the origin as an interior point.
\end{proposition}
\begin{proof}
    We have already shown that $E_S$ is convex. Since the map $\rho \to \tr({\cal O}_i \rho)$ and the space of allowed density operators is a compact set of $\mathfrak{h}$, $E_S$ is compact. 
    To show that $E_S$ has dimension $n$ and contains $0$ as an interior point, consider $\rho = \identity/N + \sum_i \epsilon_{i} {\cal O}_i$. This is Hermitian with unit trace and is positive for $\epsilon_i$ in a sufficiently small neighborhood $U$ of the origin in $\mathbb{R}^n$. We have that $\tr(\rho {\cal O}_i) = \sum_j \tr({\cal O}_i {\cal O}_j) \epsilon_j$. Since ${\cal O}_i$ are linearly independent, the matrix $M_{ij} = \tr({\cal O}_i {\cal O}_j)$ is full rank. The image of $U$ under the linear transformation $M$ thus has dimension $n$, contains 0  (the image of $\epsilon_i=0$) as an interior point, and is contained in $E_S$ by its definition. Thus, $E_S$ has dimension $n$ and contains the origin as an interior point.
\end{proof}

\section{The boundary of $E_S$ from ground states}

In this section, we provide a useful characterization of $E_S$ that provides insight into the geometry of its boundary. 

The characterization is based on the observation that for every operator \begin{equation}
    {\cal O}_{\hat{e}} = \sum_{{\cal O}_i \in S} \hat{e}_i \cdot {\cal O}_i
\end{equation}
with $\hat{e}$ a unit vector in $\mathbb{R}^n$, the expectation value of ${\cal O}_{\hat{e}}$ in any state is always larger than or equal to the least eigenvalue $\lambda_{min}(\hat{e})$ of ${\cal O}_{\hat{e}}$. This translates to the geometrical condition that $E_S$ is contained in the region $\hat{e} \cdot \vec{x} \ge \lambda_{min}(\hat{e})$. We will show that $E_S$ is the intersection of all such regions (originally proved in \cite{wichmann1963density}). 

Points on the boundary of $E_S$ are contained in hyperplanes $\hat{e} \cdot \vec{x} = \lambda_{min}(\hat{e})$; these points correspond to states that have minimum expectation value for some non-zero operator in $\text{span}(S)$, i.e. they correspond to quantum ground states for non-zero Hamiltonians in $\text{span}(S)$.

The main result is captured by Theorem \ref{boundary} stated in the introduction. We now prove this:
\begin{proof}
We begin by proving (2.).
   Consider a point $p \in \partial E_S$. Since $E_S$ is a convex region of $\mathbb{R}^n$, the Supporting Hyperplane Theorem ensures that 
     there is a ``supporting hyperplane'' $\hat{e} \cdot \vec{x} = y$ containing $p$ such that $E_S$ is entirely contained in the region $\hat{e} \cdot \vec{x} \ge y$. For any state $\rho_p$ with $\tr(\rho_p \vec{\cal O}) = \vec{x}_p$, we have $\tr(\rho_p {\cal O}_{\hat{e}}) = \hat{e} \cdot \vec{x}_p = y$ and for any other state, we have
    $\tr(\rho {\cal O}_{\hat{e}}) = \hat{e} \cdot \vec{x} \ge y$. Thus, $y$ is the minimum expectation value of ${\cal O}_{\hat{e}}$, equal to $\lambda_{min}({\cal O}_{\hat{e}})$. The hyperplane $\hat{e} \cdot \vec{x} = \lambda_{min}({\cal O}_{\hat{e}})$ is $\partial H_{\hat{e}}$, so $p \in \partial H_{\hat{e}} \cap E_S = B_{\hat{e}}$. This is a compact, convex subset of the hyperplane $\partial H_{\hat{e}}$ since it is the intersection of the closed, convex set $\partial H_{\hat{e}}$ with the compact, convex set $E_S$. 

   Since $B_{\hat{e}}$ is the set of allowed expectation values which also minimize $\tr(\rho \hat{e} \cdot \vec{\cal O})$, it is immediate that $\rho$ has expectation values in $B_{\hat{e}}$ if and only if $\rho \in G_{\hat{e}}$, giving (3.).
    
     Next, we show (1.). Let $I$ be the intersection in (\ref{first}). Then $E_S \subset I$ since for any state $\rho$ and any $\hat{e}$ the expectation value ${\cal O}_{\hat{e}}$ must be larger than the minimum eigenvalue of ${\cal O}_{\hat{e}}$. This gives $\hat{e} \cdot \vec{x} \ge \lambda_{min} ({\cal O}_{\hat{e}})$ where $x_i = \tr({\cal O}_i \rho)$ is the point in $E_S$ corresponding to $\rho$. Thus, $\vec{x} \in H_{\hat{e}}$ for all $\hat{e}$.

    Next, we show that the boundary of $E_S$ is contained in the boundary of $I$. If $\vec{x} \in \partial E_S$, we have shown that $\vec{x} \in B_{\hat{e}}$ for some $\hat{e}$ by (\ref{ESset1}), so $\vec{x} \in \partial H_{\hat{e}}$. Since $E_S \subset I$, we have that $\vec{x} \in I$. Since $\vec{x} \in \partial H_{\hat{e}}$ and $I$ is entirely contained in $H_{\hat{e}}$,  $\partial H_{\hat{e}}$ is a supporting hyperplane for $I$. Thus, $\vec{x}$ lies on the boundary of $I$.

    A topological argument now shows that $\partial E_S = \partial I$. According to Proposition \ref{closedconvex}, $E_S$ is a compact convex set of $\mathbb{R}^n$ with non-vanishing interior, so its boundary is homeomorphic to $S^{n-1}$.  The set $I$ is an intersection of closed convex sets, so is closed and convex. It is also bounded, since for any direction $\hat{e}$ and any point $\vec{x}$ in the intersection, we have $\vec{x} \in H_{\hat{e}} \cap H_{-\hat{e}}$ so $\lambda_{min}({\cal O}_{\hat{e}}) \le \hat{e} \cdot \vec{x} \le \lambda_{max}({\cal O}_{\hat{e}})$. Thus, $I$ is also a compact convex subset of $\mathbb{R}^n$. Since $I$ contains $E_S$, it also has a non-vanishing interior. It follows that the boundary of $I$ is also homeomorphic to $S^{n-1}$. Under the homeomorphism between $\partial I$ and $S^{n-1}$, $\partial E_S$ maps to some subset of $S^{n-1}$. No proper subset of $S^{n-1}$ is homeomorphic to $S^{n-1}$ (a corollary of the Borsuk-Ulam theorem), so it must be that $\partial E_S = \partial I$. The assertion (\ref{first}) now follows, since $E_S$ and $I$ are each equal to the convex hull of their boundary.

    It remains to show (4.). According to Theorem \ref{boundary}, $B_{\hat{e}}$ is the set of expectation values for the operators in $S$ when we restrict to the states in $G_{\hat{e}}$. If $\{v_a\}$ is an orthonormal basis for ${\cal H}_{\hat{e}}$, we can write the most general $\rho \in G_{\hat{e}}$ as $\rho = \sum_{a,b} \rho_{ab} v_a v_b^\dagger$ where $\rho_{ab}$ is a unit trace non-negative Hermitian matrix. Since $\pi_{\hat{e}} = \sum_a v_a v^\dagger_a$, we have that $\rho =  \pi_{\hat{e}} \rho \pi_{\hat{e}}$ in this case. Then for $\rho \in G_{\hat{e}}$,
   \begin{equation}
        \tr(\rho {\cal O}^i )  = \tr([\pi_{\hat{e}} \rho \pi_{\hat{e}}] {\cal O}^i ) = \tr( \rho [\pi_{\hat{e}} {\cal O}^i \pi_{\hat{e}}]).
    \end{equation}  
    Here, $\rho$ restricted to ${\cal H}_{\hat{e}}$ can be an arbitrary density operator on ${\cal H}_{\hat{e}}$, so the set $B_{\hat{e}}$ of such expectation values is exactly the set $E_{S_{\hat{e}}}$ for the operators in $S_{\hat{e}}$ acting on ${\cal H}_{\hat{e}}$.
\end{proof}

When ${\cal O}_{\hat{e}}$ has a non-degenerate minimum eigenvalue, $G_{\hat{e}}$ is a single state $\rho = v_{\hat{e}} v_{\hat{e}}^\dagger$ and $B_{\hat{e}}$ is a single point $\vec{x} = v_{\hat{e}}^\dagger \vec{\cal O} v_{\hat{e}}$. More generally, (4.) shows that the set $B_{\hat{e}}$ is simply the set $E_{S}$ for a reduced set of operators acting on a smaller Hilbert space, so all of the results of this paper characterizing $E_S$ can also be applied to characterize each hyperplanar face of $E_S$.

When ${\cal H}_{\hat{e}}$ has dimension $m$, the maximum dimension of $B_{\hat{e}}$ is $m^2 -1$, since this is the dimension of $G_{\hat{e}}$. The dimension of $B_{\hat{e}}$ is also bounded above by $n-1$, since $B_{\hat{e}}$ is contained in the hyperplane $\hat{e} \cdot x = \lambda_{min}({\cal O}_{\hat{e}})$. We can understand this from the fact that a linear combination $\sum_i\hat{e}_i (\pi_{\hat{e}} {\cal O}^i \pi_{\hat{e}})$ of the operators in $S_{\hat{e}}$ is proportional to the identity when acting on ${\cal H}_{\hat{e}}$, so we have at most $n-1$ linearly independent traceless operators in $span(S)$.

In the case where $B_{\hat{e}}$ is of maximum dimension $n-1$, it provides a hyperplanar face for $E_S$. The sides of the polyhedra in the example of Section \ref{sec:polygon} provide an example of this. When the dimension $2N-2$ of the space of pure states is smaller than $n-1$, the dimension of the boundary of $E_S$, only a measure 0 subset of $E_S$ can correspond to pure states, so in this case with $2N  \le n \le N^2 - 1$, generic points in $\partial E_S$ must lie in sets $B_{\hat{e}}$ corresponding to ${\cal O}_{\hat{e}}$ with degenerate eigenvalues.

\subsubsection*{Displaying $E_S$}

Theorem \ref{boundary}  gives a practical method to calculate and display $\partial E_S$ for various examples. For example, when $S$ is a pair of operators, we define $\hat{e}_\theta = (\cos(\theta), \sin(\theta))$, determine the ground state eigenvector(s) $v_\theta$ for $\hat{e}_\theta \cdot {\cal O}$ for (some discrete set of) $\theta$s  in $[0, 2 \pi]$ and calculate $\vec{x} = v_\theta^\dagger {\cal O} v_\theta$. 

As an explicit example, consider ${\cal H} = \mathbb{R}^3$ with $S = \{{\cal O}_1,{\cal O}_2\}$ and operators defined in $\ref{defops}$. Taking $\hat{e} = (\cos(\theta),\sin(\theta))$, we find that for $\theta = 0$ (i.e. $\hat{e} = (1,0)$, we have ${\cal O}_{\hat{e}} = \hat{e} \cdot {\cal O} = {\cal O}_1$ so there is a two-dimensional space of eigenvectors with minimum eigenvalue $-1$. The space $G_{\hat{e}}$ is represented by the set of density matrices of the form
\begin{equation}
    \rho =  {1 \over 2} \left( \begin{array}{ccc}
     (1 + z) & x - i y & 0 \cr
      x + i y & {1 \over 2}(1 - z) & 0 \cr
      0 & 0 & 2
    \end{array} \right)
\end{equation}
with $x^2 + y^2 + z^2 \le 1$. For these, we have $(\langle {\cal O}_1 \rangle,\langle {\cal O}_2 \rangle) = (-1,x)$ so $B_{\hat{e}}$ is the line segment $\{(-1,x) | |x| \le 1\}$. For all other $\hat{e}$, the minimum eigenvalue is non-degenerate with some normalized eigenvector $v_\theta$, and $B_{\hat{e}}$ is the single point corresponding to $\vec{x} = v_\theta^\dagger {\cal O} v_\theta$. These points trace out a smooth curve connecting $(-1,-1)$ with $(-1,1)$ as $\theta$ increases from $0$ to $2 \pi$, as shown in Figure \ref{fig:example}. The fact that $\lim_{\theta \to 0^+} B_{\hat{e}(\theta)} \ne \lim_{\theta \to 0^-} B_{\hat{e}(\theta)}$ is related to the familiar fact that the eigenvector with minimum eigenvalue can change discontinuously as we pass through a point where the minimum eigenvalue is degenerate. 

\section{The interior of $E_S$ from thermal states}
\label{sec:boundary}

In the previous section, we have understood that the boundary points of $E_S$ are realized by ground states of Hamiltonians $\hat{e} \cdot {\cal O}$ formed from linear combinations the operators in $S$.
In this section, we will see that there is actually an isomorphism between the $T>0$ thermal states of these Hamiltonians (reviewed below) and the interior of $E_S$. These thermal states for Hamiltonians in $span(S)$ are generally a much-restricted set of states compared to the full set of density operators (if $dim(S) \ll N^2-1$). This leads to an efficient way to determine whether a certain set $\vec{x}$ of expectation values is in $E_S$ and if so to construct a state that realizes these expectation values.

We begin with a few definitions.
\begin{definition}
    For a set $S = \{{\cal O}_i\}$ of $n$ Hermitian operators, we define the partition function $Z_S:\mathbb{R}^n \to \mathbb{R}$ by
    \begin{equation}
    Z_S(\vec{\beta}) = \tr e^{-\beta_i {\cal O}_i} \; ,
\end{equation}
the thermal state $\rho_{\vec{\beta}}$ by
\begin{equation}
\label{thermal}
    \rho_{\vec{\beta}} = {1 \over Z_S(\vec{\beta}) }e^{-\beta_i {\cal O}_i} \; ,
\end{equation}
and the thermal expectation value map $E_S:\mathbb{R}^n \to \mathbb{R}^n$ by 
\begin{equation}
\label{Ebeta}
E_i(\vec{\beta}) = \tr({\cal O}_i \rho_{\vec{\beta}}) =  -{\partial \over \partial \beta_i} \ln Z_S(\vec{\beta}).
\end{equation}
\end{definition}
To make contact with ordinary thermodynamics, 
these quantities are the thermal partition function, thermal state, and thermal expectation values for a quantum system with temperature $1/|\vec{\beta}|$ and Hamiltonian $H = \hat{e} \cdot \vec{\cal O}$ where $\hat{e} = \vec{\beta}/\beta$. The limits $\beta \to \infty$ of these states give some of the ground states of the previous section.\footnote{Specifically, they give the ground states with non-degenerate least eigenevalue and states $\rho = \sum_i v_i v^\dagger_i /m$ where $v_i$ is any orthonormal basis of the degenerate ground states otherwise.} 

A central result (proved originally in \cite{wichmann1963density}) is:\footnote{The slightly different version in the introduction follows immediately from this, since for a general set of linearly independent operators with $\identity \notin \text{span}(S)$, we can write ${\cal O}_i = {\cal O}_i' + c_i \identity$ with $\{{\cal O}_i'\}$ linearly independent and traceless. The map $\vec{\beta} \to \vec{E}(\vec{\beta})$ for $\{{\cal O}_i'\}$ is just the map $\vec{\beta} \to \vec{E}(\vec{\beta})$ for $\{{\cal O}_i\}$ composed with translations.}

\noindent
{\bf Theorem 2.} {\it
If $S = \{{\cal O}_i \}$  is a set of $n$ linearly independent traceless operators, the map $\vec{\beta} \to \vec{E}(\vec{\beta})$ 
is a diffeomorphism between $\mathbb{R}^n$ and the interior of $E_S$ with symmetric negative definite Jacobian $\partial E_i / \partial \beta_j = \partial_i \partial_j \ln Z(\vec{\beta})$ .
}

This map $\beta \to \vec{E}(\vec{\beta})$ provides a canonical set of coordinates $\vec{\beta}$ on the interior of $E_S$. The theorem has various immediate physical consequences. We have 
\begin{corollary}
    For $S = \{{\cal O}_i \}$ a set of $n$ linearly independent traceless operators and $\vec{E}$ a set of allowed expectation values in the interior of $E_S$ with $\vec{E} \ne 0$, there is a unique Hamiltonian $H = \hat{e} \cdot \vec{\cal O}$ ($|\hat{e}| = 1$)  and a temperature $T$ such that the thermal state of the quantum system with Hamiltonian $H$ and temperature $T$ has expectation values $\vec{E}$ for the operators in $S$.
\end{corollary}
\begin{proof}
    Since the map (\ref{Ebeta}) is a diffeomorphism, it is one-to-one and onto. Thus, we have a bijection between $\vec{E} \in int(E_S)$ and vectors $\vec{\beta} \in \mathbb{R}^n$. We also have a bijection between $\vec{\beta}$ and pairs $(T,\hat{e})$ via $T = 1/|\vec{\beta}|$ and $\hat{e} = \vec{\beta}/ |\vec{\beta}|$.
\end{proof}
The positivity of the Jacobian implies that for Gibbs states, energy always increases with temperature, and the expectation value of an operator ${\cal O}$ always increases with the chemical potential for that operator:
\begin{corollary}
    In any quantum system with Hamiltonian $H$ and some Hermitian operator ${\cal O}$, the Gibbs state  $\rho = {\cal N} \exp(-\beta H + \mu {\cal O})$ satisfies 
    \begin{equation}
   {d \langle {\cal O} \rangle \over d \mu} > 0 \qquad \qquad -{d \langle H \rangle \over d \beta} > 0.
    \end{equation}
\end{corollary}
\begin{proof}
    The two assertions are special cases of the more general observation that 
    \begin{equation}
        {d \over d \epsilon} \tr(\sum_j c_j {\cal O}_j \rho_{\vec{\beta} - \epsilon \vec{c}})|_{\epsilon=0} = -\sum_{ij} c_j \partial_i E_i c_j >0
    \end{equation}
    which follows since $-\partial_i E_j$ is positive definite. 
\end{proof}
The result implies a positivity condition on the connected thermal two-point function for arbitrary operators:
\begin{corollary}
    In a quantum system with Hamiltonian $H$ and some Hermitian operator ${\cal O}$, define $\langle {\cal O} \rangle_\beta = \tr(\rho_\beta {\cal O})$, $\Delta {\cal O} = {\cal O} - \langle {\cal O} \rangle_\beta$ where $\rho_\beta$ is the thermal state $e^{- \beta H}/\tr(e^{-\beta H})$, and ${\cal O}(\tau) = e^{H \tau} {\cal O} e^{-H \tau}$. Then 
    \begin{equation}
    \int_0^\beta d \tau \langle \Delta {\cal O}(\tau ) \Delta {\cal O}(0) \rangle_{\beta} > 0.
    \end{equation}
\end{corollary}
\begin{proof}
    Consider $S$ with ${\cal O}_1 = H$, and ${\cal O}_2 = {\cal O}$ with $\vec{\beta} = (\beta,0)$. Let $\vec{v} = (0,1)$. Then the positivity of $-\partial_i E_j$ at $\vec{\beta}$ implies
    \begin{eqnarray*}
        0 &<& - v_i \partial_i E_j v_j \cr
        &=& \left. -\partial_\epsilon { \tr({\cal O} e^{- \beta H + \epsilon {\cal O}}) \over  \tr(e^{- \beta H + \epsilon {\cal O}})} \right|_{\epsilon = 0} \cr
        &=& -{ \int_0^1 ds \tr({\cal O} e^{- (1-s) \beta H} {\cal O} e^{- s \beta H}) \over \tr(e^{-  \beta H})} + { \tr^2({\cal O}  e^{- \beta H}) \over \tr^2(e^{-  \beta H})} 
    \end{eqnarray*}
    where we have used ${d \over d \epsilon} e^{A + \epsilon B} = \int_0^1 ds e^{(1-s)A} B e^{s A}$. Using the definitions, we find that $\beta$ times the final line is $\int_0^\beta d \tau \langle \Delta {\cal O}(\tau ) \Delta {\cal O}(0) \rangle_{\beta}$. 
\end{proof}
Finally, we have a concavity result for the logarithm of the partition function:
\begin{corollary}
    For any set $S$ of linearly independent traceless operators, the function $\ln Z(\vec{\beta}) = \ln \tr(e^{-\beta_i {\cal O}_i})$ is concave.
\end{corollary}
\begin{proof}
    We have that $\partial_i \ln Z = -E_i (\vec{\beta})$ so the Hessian $\partial_i \partial_j f = -\partial_i E_j$ is negative definite.
\end{proof}

We now proceed to the proof of Theorem \ref{expthm}.

\subsubsection*{Physical motivation}

We begin with the following physical motivation for the result \cite{wichmann1963density}:
\begin{lemma}
    If there is a state $\rho$ with expectation values $\vec{E}$ for the operators in $S$, then there is a state $\rho_{\vec{E}}$ of maximum von Neumann entropy $S_{vN} = -\tr(\rho \log \rho)$ with the same expectation values.
\end{lemma}
\begin{proof}
By our assumption, there is at least one state with expectation values $\vec{E}$. The set of all such states is compact since it is closed (if $\rho$ is the limit of a sequence of operators with expectation values $\vec{E}$, then the expectation values for $\rho$ are also $\vec{E}$ by the continuity of the map $\rho \to \tr(\rho \vec{O})$) and a subset of the compact set of density operators. The von Neumann entropy is a continuous function on this compact set, so it must have a maximum. 
\end{proof}
The connection to the theorem above comes from the following lemma:
\begin{lemma}
Suppose $\rho$ lies in the interior of the space of density operators and maximizes $S_{vN} = -\tr(\rho \log \rho)$ subject to the constraints $\tr(\rho {\cal O}_i) = E_i$ for ${\cal O}_i \in S$ together with the required $\tr(\rho) = 1$. Then 
\begin{equation} 
\rho = {e^{-\beta_i {\cal O}_i} \over \tr(e^{-\beta_i {\cal O}_i})}
\end{equation}
for some $\vec{\beta} \in \mathbb{R}^n$.
\end{lemma}
\begin{proof}
The method of Lagrange multipliers shows that the action
\begin{equation}
    I =  -\tr(\rho \log \rho) - \beta_i (\tr(\rho {\cal O}_i) - E_i) + \beta_0 (\tr(\rho) -1)
\end{equation}
is stationary with respect to variations of $\rho$. This gives
\begin{equation}
      - \log \rho - \beta_i  {\cal O}_i  - \beta_0  = 0
\end{equation}
so $\rho = \exp(-\beta_0)\exp(-\beta_i  {\cal O}_i)$. The condition $\tr(\rho) = 1$ implies $\exp(\beta_0) = \tr(\exp(-\beta_i {\cal O}_i))$.
\end{proof}
According to the lemmas, any set of allowed expectation values can be reproduced by a state that maximizes von Neumann entropy state subject to these expectation values, and such states take the form (\ref{thermal}) provided that their density matrix is non-singular. This implies that we can find a thermal state to reproduce any set of expectation values except possibly when these only arise from degenerate states. The theorem is stronger, since it implies that any expectation values in the interior of $E_S$ can be produced by some unique state of the form (\ref{thermal}), and the map from $\vec{\beta}$ to $\vec{E}$ is a diffeomorphism (differentiable with a differentiable inverse) with a negative definite Jacobian.

To complete the proof of the stronger result Theorem (\ref{expthm}), it will be useful to give a more geometrical description of the map (\ref{Ebeta}).

\subsubsection*{Geometrical description}

To begin, we recall that $\mathfrak{h}$ is the set of Hermitian operators on ${\cal H}$ and define $\mathfrak{h}^*$ to be the dual vector space of linear maps from $\mathfrak{h}$ to $\mathbb{R}$. For a set $S  =\{ {\cal O}_i\}$ of $n$ linearly independent operators, we define the subspace $\mathfrak{h}_S \subset \mathfrak{h}$ to be $span(S)$ and $\mathfrak{h}^*_S$ to be the dual vector space. The map $\vec{E}(\vec{\beta})$ in Theorem \ref{expthm} can be understood as a composition 
\begin{equation}
\mathbb{R}^n \xrightarrow{C_S} \mathfrak{h}_S \xrightarrow{\iota_S} \mathfrak{h} \xrightarrow{df} \mathfrak{h}^* \xrightarrow{\iota^*_S} \mathfrak{h}^*_S \xrightarrow{{\cal E}_S} \mathbb{R}^n \; .
\end{equation}
Here,
\begin{itemize}
    \item 
$C_S$ is the isomorphism $C_S(\vec{\beta}) = -\sum \beta_i {\cal O}_i$ whose inverse assigns coordinates to elements in $\mathfrak{h}_S$ for the basis $\{ -{\cal O}_i\}$.
\item $\iota_S$ is the inclusion map.
\item $df$ is the derivative of 
\begin{equation}
f({\cal O}) = \log \tr e^{\cal O} \; ,
\end{equation}
given by 
\begin{equation}
df({\cal O}) = {\tr( \cdot \; e^{\cal O}) \over \tr(e^{\cal O})} \; .
\end{equation}
\item 
$\iota^*_S$ is the restriction map defined by restricting the action of map in $\mathfrak{h}^*$ to the subspace $\mathfrak{h}_S$.
\item 
${\cal E}_S$ is defined by $[{\cal E}_S({\cal O})]_i = \tr({\cal O} {\cal O}_i)$. This is map is the isomorphism that assigns coordinates to ${\cal O} \in \mathfrak{h}_S$ in the dual basis $\hat{\cal O}_i$ defined by $\hat{\cal O}_i \cdot {\cal O}_j = \delta_{ij}$.
\end{itemize}
The set of valid quantum states is a subset ${\cal D}^* \subset \mathfrak{h}^*$ defined by maps ${\cal O} \to \tr(\rho {\cal O})$ where $\rho$ is a non-negative unit trace Hermitian operator. Restricting the maps in ${\cal D}^*$ to their action on $\mathfrak{h}_S$ defines ${\cal D}^*_S = \iota^*_S {\cal D}^* \subset \mathfrak{h}^*_S$. 

Since $e^{\cal O}/\tr(e^{\cal O})$ is positive and has unit trace, the image of $\mathfrak{h}$ under $df$ lies in ${\cal D}^*$ and the image of $\mathfrak{h}_S$ under $\iota^*_S \circ df \circ \iota_S$ lies in ${\cal D}^*_S$. The content of Theorem \ref{expthm} amounts to a strengthening of the latter statement:

\noindent
{\bf Theorem 3}
{\it
For a subspace $\mathfrak{h}_S$ of the space of traceless Hermitian operators, the map $df_S = \iota^*_S \circ df \circ \iota$ is a diffeomorphism from $\mathfrak{h}_S$ to the interior of ${\cal D}^*_S$ with symmetric, positive definite Jacobian.}

To see the equivalence of this statement with Theorem \ref{expthm}, we note that the set $E_S$ of allowed expectation values for the operators in $S$ is by our definitions the image of ${\cal D}^*_S$ under ${\cal E}_S$, so Theorem \ref{expthm} is exactly Proposition 1 expressed in coordinates.

\begin{figure}
    \centering
    \includegraphics[width = .7 \textwidth]{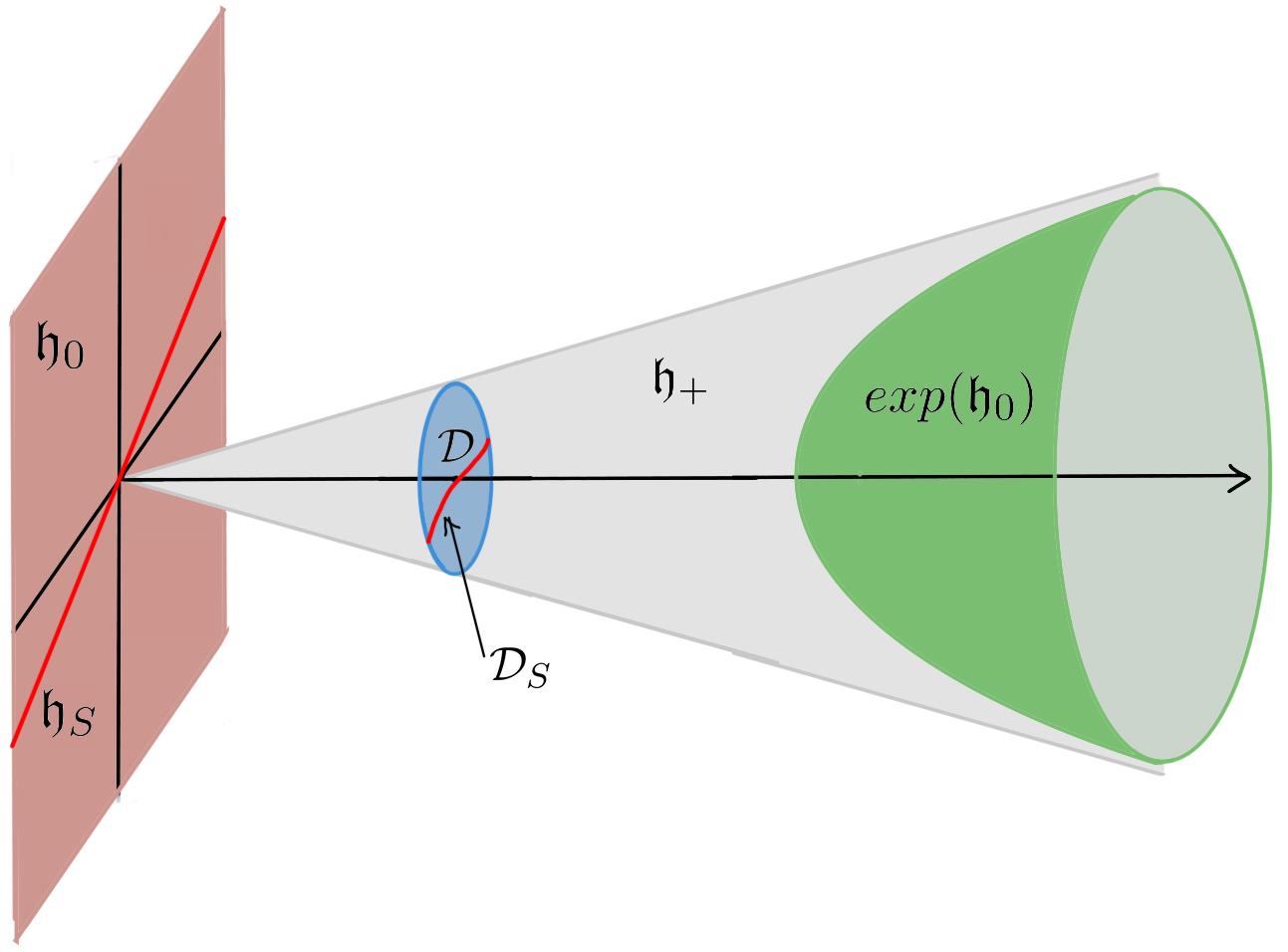}
    \caption{Space $\mathfrak{h}$ of Hermitian operators, with the direction corresponding to the identity operator pointing to the right. The exponential map is a diffeomorphism from $\mathfrak{h}$ to positive Hermitian operators $\mathfrak{h}_+$. The image (shown in green) of $exp$ on the traceless Hermitian operators $\mathfrak{h}_0$ maps onto the space ${\cal D}_+$ of positive density operators (blue) under the normalization map ${\cal N}:{\cal O} \to {\cal O}/\tr({\cal O})$. For a subspace $\mathfrak{h}_S \subset \mathfrak{h}_0$, the image of ${\cal N} \circ \exp$ is a subset ${\cal D}_S \subset {\cal D}_+$. For any allowed set of expectation values in the interior of $E_S$, there is a unique state in ${\cal D}_S$ that achieves these expectation values.}
    \label{fig:mathpic}
\end{figure}

To understand better the map $df$ and set up our proof below, it will be useful to further decompose $df$ into a composition
\begin{equation}
\mathfrak{h} \xrightarrow{df} \mathfrak{h}^* \qquad = \qquad  \mathfrak{h} \xrightarrow{exp} \mathfrak{h}_+ \xrightarrow{\cal N} \mathfrak{h} \xrightarrow{i} \mathfrak{h}^* \; .
\end{equation}
Here, the exponential map $exp: {\cal O} \to e^{\cal O}$ takes the space $\mathfrak{h}$ of Hermitian operators on ${\cal H}$ to the space $\mathfrak{h}_+$ of positive Hermitian operators, which is a convex cone in $\mathfrak{h}$. The normalization map ${\cal N} : {\cal O} \to {\cal O}/\tr({\cal O})$ takes the space of positive Hermitian operators to the space ${\cal D}_+$ of positive density operators whose closure is the full space ${\cal D}$ of density operators. Finally, $i$ is the canonical isomorphism ${\cal O} \leftrightarrow \tr(\cdot \; {\cal O})$.

Figure \ref{fig:mathpic} provides a visualization of the action of $\exp$ and ${\cal N}$ on $\mathfrak{h}$ and on a subspace $\mathfrak{h}_S$.

As we show below, the composition ${\cal N} \circ exp$ is a diffeomorphism from the space $\mathfrak{h}_0$ of traceless Hermitian operators to the space of positive density operators. Thus, each positive density operator arises from some unique traceless operator under the map. Restricting ${\cal N} \circ exp$ to the subspace $\mathfrak{h}_S$, we get a diffeomorphism between $\mathfrak{h}_S$ and a subset ${\cal D}_S$ of the positive density operators. Theorem \ref{dualprop} states that for any way of assigning expectation values to the operators in $\mathfrak{h}_S$ associated with a valid quantum state in ${\cal D}_+$, there is some unique state in ${\cal D}_S$ that makes the same assignments. Further, that the full map between $\mathfrak{h}_S$ and ${\cal D}_S$ is a diffeomorphism with positive Jacobian. We are now ready for the proof of Theorems \ref{expthm} and \ref{dualprop}.

\subsection{Proof of Theorems \ref{expthm} and \ref{dualprop}}

We begin with
\begin{lemma}\label{lemma:exp_diff}
    The exponential map is a diffeomorphism between the real vector space $\mathfrak{h}$ of Hermitian operators acting on an $N$-dimensional Hilbert space and the set of positive Hermitian operators $\mathfrak{h}_+ \subset \mathfrak{h}$.
\end{lemma}
\begin{proof}
    We make use of the standard isomorphism between $\mathfrak{h}$ and the real vector space of $N\times N$ Hermitian matrices viewed as a smooth manifold. Let $H^+$ be the set of positive definite Hermitian matrices. Since positivity is an open condition, $H^+$ has the structure of a smooth manifold. Let $\text{exp}:H\to H^+$ be the exponential map. Smoothness of $\text{exp}$ follows from the uniform convergence of the series expansion of the exponential. 
    
    We first show that $\text{exp}$ is injective. Suppose we have  $A,B\in H$ such that $e^A = e^B$. Since $A$ and $B$ are Hermitian, they are unitarily diagonalizable. Further, $A$ and $B$ must have the same eigenvalues, since if $\{\lambda_i\}$ are the eigenvalues of $H$, the eigenvalues of $e^H$ are $\{e^{\lambda_i}\}$ and the exponential function is invertible. Thus, can write $A = U D U^{-1}$ and $B = W D W^{-1}$ where $D$ is a diagonal matrix with diagonal entries ordered from smallest to largest. Now, $e^A = e^B$ gives $U e^{D} U^{-1} = W e^{D} W^{-1}$ or $[W^{-1} U, e^{D}] = 0$. The vanishing of this commutator implies that $W^{-1} U$ is block diagonal, with blocks corresponding the sets of rows for which $e^D$ has equal eigenvalues. Since $D$ has equal eigenvalues for these same sets of rows, we also have $[W^{-1} U, D] = 0$ or $U D U^{-1}= W D W^{-1}$. Thus, $A=B$.

    We now show surjectivity. Let $B\in H^+$. Then there exists a unitary operator $U$ such that $B=UDU^{-1}$ for some diagonal $D$ with positive, real entries. Since the exponential is a bijection from the reals to positive reals, we can find $\lambda_i\in\mathbb{R}$ such that $e^{\lambda_i}=d_i$ where $d_i$ is the $i$th diagonal of $D$. We can then construct $A = U\text{diag}(\lambda_1,...,\lambda_N)U^{-1}\in H$ and observe that 
    \begin{align*}
       e^A = U\text{diag}(e^{\lambda_1},...,e^{\lambda_N})U^{-1}=U\text{diag}(d_1,...,d_N)U^{-1} = B.
    \end{align*}
    Hence we have that $\text{exp}$ is surjective and by the previous paragraphs, it is a smooth bijection.

    It now suffices to show that for $\phi = \text{exp}$, the derivative $d\phi_X$ is invertible for each $X\in H$. This implies that $\phi$ is a local diffeomorphism by the inverse function theorem and hence a global diffeomorphism as we have already shown that $\phi$ is a smooth bijection. We have that \cite{rossmann2002lie}
    \begin{equation}
        d\phi_X(Y) = \left. {d \over d \epsilon} e^{X + \epsilon Y} \right|_{\epsilon = 0} = \int ds \; e^{(1-s)X} Y e^{sX} = e^X {1 - e^{-ad(X)} \over ad{X}} [Y].
    \end{equation}
    where $ad(X):M \to [X, M]$. 
    We will show that all the eigenvalues of this map (acting on the full space of $N \times N$ matrices) are all positive. Let $v_i$ be an orthonormal basis of eigenvectors for $X$ with corresponding eigenvalues $\lambda_i$. Then $v_i v_j^\dagger$ provide an orthonormal basis for $N \times N$ matrices and 
    \begin{equation}
        d\phi_X(v_i v_j^\dagger) = \int ds \; e^{(1-s)X} v_i v_j^\dagger e^{sX} = \int ds \; (e^{(1-s)\lambda_i}  e^{s \lambda_j}) v_i v_j^\dagger \; .
    \end{equation}
    Thus, $v_i v_j^\dagger$ is an eigenvector for $d\phi_X$ and the eigenvalue $\int ds \; (e^{(1-s)\lambda_i}  e^{s \lambda_j})$ is clearly positive.  
\end{proof}
Next, we show that 
\begin{lemma}
\label{posjac}
    The map ${\cal N} \circ \exp$ is a diffeomorphism from the space $\mathfrak{h}_0$ of traceless Hermitian operators to the set ${\cal D}_+$ of positive density operators, with symmetric positive definite Jacobian.
\end{lemma}
\begin{proof}
For a positive operator ${\cal O}$, we have that $\tr({\cal O})> 0$, so the normalization map ${\cal N}:{\cal O} \to {\cal O}/\tr({\cal O})$ is clearly smooth on $\mathfrak{h}_+$. For any ${\cal O} \in \mathfrak{h}_0$, the image $\text{exp}({\cal O})/\tr(\text{exp}({\cal O}))$ is positive, Hermitian, and trace 1, so is in ${\cal D}_+$. Thus, the composition ${\cal N} \circ \exp$ is a smooth map from $\mathfrak{h}_0$ to ${\cal D}_+$. Next, a simple calculation shows that the map
\begin{align*}
     \rho \mapsto \frac{\rho}{e^{\frac{1}{N}\tr(\log(\rho))}}
\end{align*}
${\cal D}_+$ to $\mathfrak{h}_0$ is an inverse for ${\cal N} \circ \exp$. This is smooth, since it is a composition of smooth maps. Thus ${\cal N} \circ \exp$ is a smooth map from $\mathfrak{h}_0$ onto ${\cal D}_+$ with a smooth inverse, so it is a diffeomorphism. To show that the Jacobian is positive, note that for any ${\cal O} \in \mathfrak{h}_0$, we have
\begin{equation}
    \left.{d \over d \epsilon}{\cal N} \circ \exp(\epsilon {\cal O})\right|_{\epsilon = 0} = \left. {d \over d \epsilon} {e^{\epsilon {\cal O}} \over \tr (e^{\epsilon {\cal O}})} \right|_{\epsilon = 0} = {1 \over N } {\cal O} 
\end{equation}
so the Jacobian evaluated at the origin of $\mathfrak{h}_0$ is proportional to the identity map, and therefore positive. Since ${\cal N} \circ \exp$ is a smooth diffeomorphism, its determinant is continuous and non-zero everywhere, so must be everywhere positive. 
\end{proof}
Now, we restrict the domain to $\mathfrak{h}_S$ and consider the map $df_S = \iota^*_S \circ i \circ {\cal N} \circ \exp \circ \iota_S$. We have:
\begin{lemma}
    For a subspace $\mathfrak{h}_S \subset \mathfrak{h}_0$, the map $df_S$ is a diffeomorphism from $\mathfrak{h}_S$ to its image in $\mathfrak{h}_S^*$ with positive Jacobian.
\end{lemma}
\begin{proof}
    Let $\{{\hat{\cal O}}_i\}$ ($i = 1, \dots,N^2-1$) be an orthonormal basis for $\mathfrak{h}_0$ such that the restriction to ($i = 1, \dots,n$) gives an orthonormal basis for $\mathfrak{h}_S$.   Using the coordinate representation associated with these bases, the map ${\cal N} \circ \exp$ on $\mathfrak{h}_0$ is represented by the map $y: \mathbb{R}^{N^2-1} \to \mathbb{R}^{N^2-1}$ with
    \begin{equation}
        y^i(\vec{x}) = {\tr(\hat{{\cal O}}_i e^{x^j \hat{{\cal O}}_j}) \over \tr( e^{x^j \hat{{\cal O}}_j})}.
    \end{equation}
    By Lemma \ref{posjac} the Jacobian $J_{ij} = \partial y_i / \partial x_j$ is positive. The map $df_S: \mathfrak{h}_S \to \mathfrak{h}_S^*$ is represented by the same formula, with the basis elements restricted to $i=1 \dots n$. The Jacobian $J_S$ is the upper $n \times n$ block of the Jacobian $J$ matrix. Given any $\vec{v} \in \mathbb{R}^n$, we can define $\vec{v}_+ = (\vec{v}, \vec{0}) \in \mathbb{R}^{N^2-1}$ so that $\vec{v}^T J_S \vec{v} = \vec{v}_+^T J \vec{v}_+ > 0$. Since the Jacobian is positive definite, it is invertible everywhere, so  $df_S$ is a local diffeomorphism. To see that the function is a bijection, we note that for $\Delta \vec{x} \equiv \vec{x}_2 - \vec{x}_1 \ne 0$, 
    \begin{equation}
       \Delta \vec{x} \cdot(\vec{y}(\vec{x}_2) - \vec{y}(\vec{x}_1)) = \Delta \vec{x} \cdot \int_0^1 ds \partial_s \vec{y}(\vec{x}_1 + s \Delta \vec{x}) = \int_0^1 ds \Delta \vec{x} \cdot [J_S (\vec{x}_1 + s\Delta \vec{x}) \Delta \vec{x}] > 0
    \end{equation}
    so $\vec{y}(\vec{x}_2) \ne \vec{y}(\vec{x}_1)$.
\end{proof}
To complete the proof of Theorem \ref{expthm} and \ref{dualprop}, we need to show that the image of $\mathfrak{h}_S$ under $df_S$ is the interior of ${\cal D}_S^*$, i.e. that every set of expectation values in the interior of $E_S$ is attained by a state of the form \ref{thermal} with ${\cal O}_i \in S$:
\begin{lemma}
The map ${\cal E}_S \circ df_S$ on $\mathfrak{h}_S$ is surjective onto the interior of $E_S$.
\end{lemma}
\begin{proof} 
For $S = \{{\cal O}_i\}$, the space $\mathfrak{h}_S = span(S)$ consists of the zero operator together with operators of the form 
\begin{equation}
    {\cal O}(\beta, \hat{e}) = - \beta \hat{e} \cdot \vec{\cal O}  
\end{equation}
where $\hat{e}$ is a unit vector in $\mathbb{R}^n$ and $\beta \in (0,\infty]$. For fixed $\beta > 0$, this set of operators (which we denote by $\mathfrak{h}_\beta$) has the topology of a sphere $S^{n-1}$ in $\mathfrak{h}_S$. We define $E_{\beta} \in E_S$ to be the image of $\mathfrak{h}_\beta$ under the diffeomorphism ${\cal E}_S \circ df_S$ and $B_{\beta}$ to be the image of the ball bounded by $\mathfrak{h}_\beta$. To prove the lemma, we show that any point in the interior of $E_S$ lies in $B_\beta$ for sufficiently large $\beta$.

First, since the zero operator is inside the sphere $\mathfrak{h}_\beta$ and maps to $0 \in \mathbb{R}^n$ under ${\cal E}_S \circ df_S$, we have that 0 is an interior point of $B_\beta$. To show that $\lim_{\beta \to \infty} B_\beta = \text{int}(E_S)$, we will demonstrate that for any $\epsilon > 0$, there is some $\beta_\epsilon$ such that for $\beta > \beta_\epsilon$, every point in  $E_{\beta}$ is within $\epsilon$ of the boundary of $E_S$. 

In this case, for any point $P$ in the interior of $E_S$, the line segment $OP$ is in the interior of $E_S$ by the convexity of $E_S$. Letting $\Delta$ be the minimum distance from $OP$ to the boundary of $E_S$, we can choose $\epsilon = \Delta$ and $\beta > \beta_\epsilon$ to ensure that every point in $E_{\beta}$ is within $\epsilon$ of the boundary of $E_S$. Since $0 \in B_\beta$, $P$ must also be in $B_\beta$, otherwise $E_\beta$ would intersect $OP$, which is impossible since all points in $OP$ have distance $> \Delta$ to the boundary of $E_S$ and all points in $E_\beta$ have distance less than $\Delta$ to the boundary.

To complete the proof, let $R$ be the range of the set of all eigenvalues for all operators in $S$ and let
\begin{equation}
\beta_\epsilon = {2 \over \epsilon} \ln \left({ 2 N R \over \epsilon} \right) \; ,
\end{equation}
and suppose $\beta > \beta_\epsilon$.

Each point $\vec{x}$ in $E_\beta$, corresponds to some $\hat{e}$. For this $\hat{e}$, let $\{\lambda_i\}$ be the eigenvalues of $\hat{e} \cdot \vec{O}$ and $\lambda_{min}$ be the minimum eigenvalue. We recall from section \ref{sec:boundary} that $\hat{e} \cdot x = \lambda_{min}$ is a supporting hyperplane for $E_S$. We will show that for $\beta > \beta_\epsilon$, $\vec{x}$ is within $\epsilon$ of this supporting hyperplane. Since $\vec{x}$ is in the interior of $E_S$, and the nearest point $P$ on the supporting hyperplane is not in the interior of $E_S$, the perpendicular segment between $x$ and $P$ must contain a point on the boundary of $E_S$. The distance between $x$ and this boundary point must be less than $\epsilon$.

To show that $\vec{x} = {\cal E}_S \circ df_S({\cal O}(\beta,\hat{e}))$ is within $\epsilon$ of the supporting hyperplane, we demonstate that $|\hat{e} \cdot \vec{x} - \lambda_{min}| < \epsilon$. We have
\begin{eqnarray*}
|\hat{e} \cdot \vec{x} - \lambda_{min}| &=&\left|{\tr(\hat{e} \cdot \vec{O} e^{ -\beta \hat{e} \cdot \vec{O}}) \over \tr(e^{ -\beta \hat{e} \cdot \vec{O}})} -\lambda_{min} \right| 
\cr
&=& \left| \sum_i (\lambda_i- \lambda_{min}) {e^{- \beta \lambda_i} \over \sum_j e^{- \beta \lambda_j}} - \lambda_{min} \right| \cr
&=& \left| \sum_{\lambda_i < \lambda_{min}  +\epsilon/2 } (\lambda_i- \lambda_{min}) {e^{- \beta \lambda_i} \over \sum_j e^{- \beta \lambda_j}} + \sum_{\lambda_i \ge \lambda_{min}  +\epsilon/2 } (\lambda_i- \lambda_{min}) {e^{- \beta \lambda_i} \over \sum_j e^{- \beta \lambda_j}} \right| \cr
&<& \left| {\epsilon \over 2} \sum_{\lambda_i < \lambda_{min}  +\epsilon/2 }  {e^{- \beta \lambda_i} \over \sum_j e^{- \beta \lambda_j}} + R \sum_{\lambda_i \ge \lambda_{min}  +\epsilon/2 } {e^{- \beta (\lambda_{min}  +\epsilon/2))} \over \sum_j e^{- \beta \lambda_j}} \right| \cr
&\le& (\epsilon/2 + R e^{- \beta \epsilon/2} N) \cr
&<& \epsilon.
\end{eqnarray*}
\end{proof}

\section{The inverse problem}

In this section, we consider the following practical question
\begin{question}
    Given a set $S$ of n operators ${\cal O}_i$, and $\vec{e} \in \mathbb{R}^n$, determine whether $\vec{e} \in E_S$ and if so, determine the most general state $\rho$ with these expectation values.
\end{question}
Making use of the results of the previous section, we provide a useful approach to answering this. According to Theorem \ref{expthm}, $\vec{e} \in \text{int}(E_S)$ if and only if there there is a state $\rho_{\vec{\beta}}$ of the form (\ref{thermal}) with expectation values $\vec{e}$, so we can proceed by looking for such a state. 

An efficient approach is to consider the flow in the space of $\vec{\beta}$s associated with the negative gradient flow for the function $\Delta(\vec{E}) = (\vec{E} - \vec{e})^2/2$ in the space of expectation values. That is, we would like to define a flow $\partial_t \vec{\beta} = \vec{v}(\vec{\beta})$ such that 
\begin{equation}
    \label{Eflow}
    \partial_t E_i = -\partial_i \Delta= -( E_i - e_i)
\end{equation} 
and thus $\vec{E}(t) = \vec{e} + (\vec{E}_0 - \vec{e})e^{-t}$. This leads to:
\begin{theorem}
Consider the flow resulting from the equation 
\begin{equation}
\label{flow}
\partial_t \beta_i = -(J^{-1})_{ji} (E_j(\vec{\beta}) - e_j) \quad \text{with}\quad J_{ij} = \partial_i E_j \; .
\end{equation}
with some arbitrary initial $\vec{\beta}$. Then along the flow, we have $\Delta(\vec{E}(\vec{\beta}(t))) = \Delta_0 e^{-2t}$ and
\begin{enumerate}
    \item
    If $\vec{e} \in int(E_S)$,
the flow converges to the unique $\vec{\beta} = \vec{\beta}_\infty$ with $\vec{E}(\vec{\beta}_\infty) = \vec{e}$. 
\item
If $\vec{e} \in \partial E_S$, $|\vec{\beta}(t)|$ diverges for $t \to \infty$ and $\lim_{t \to \infty} \Delta(\vec{E}(\vec{\beta}(t))) = 0$.
\item
If $\vec{e} \notin \partial E_S$, $|\vec{\beta}(t)|$ diverges at some finite time $t_0$ and $\lim_{t \to t_0} \Delta(\vec{E}(\vec{\beta}(t))) > 0$.
\end{enumerate}
\end{theorem} 
\begin{proof}
    Since $\vec{E}(\vec{\beta})$ is a diffeomorphism, it has an inverse $\vec{\beta}(\vec{E})$ with 
    \begin{equation}
    \label{floweq}
        {\partial \beta_i \over \partial E_j }{\partial E_j \over \partial \beta_k } = \delta_{ik}.
    \end{equation}
    Using this, it is straightforward to check that 
    \begin{equation}
        \label{betasol}
        \vec{\beta}(t) = \vec{\beta}[\vec{e} + (\vec{E}(\vec{\beta}_0)- \vec{e})e^{-t}]
    \end{equation}
    solves the flow equation (\ref{flow}) and thus gives the solution to the initial value problem with $\vec{\beta}(0) = \vec{\beta}_0$. For this solution, $\vec{E}(t) = \vec{e} + (\vec{E}_0 - \vec{e})e^{-t}$ as desired, where $\vec{E}_0 = \vec{E}(\vec{\beta}_0)$, so we have $\Delta(\vec{E}(\vec{\beta}(t))) = \Delta_0 e^{-2t}$. Now consider the various cases in the proposition.
    \begin{enumerate}
        \item 
        For any initial point $\vec{E}_0$, the trajectory of the flow (\ref{Eflow}) is the linear path from $\vec{E}_0$ to $\vec{e}$. Since $E_S$ is convex, if $\vec{e} \in int(E_S)$ and $\vec{E}_0 \in int(E_S)$, the entire flow $\vec{E}(t)$ lies within $int(E_S)$. The solution (\ref{betasol}) is the image of this flow under the diffeomorphism $\vec{\beta}(\vec{E}):int(E_S) \to \mathbb{R}^n$, so is well-defined for all $t$ and approaches $\vec{\beta}(\vec{e}) \equiv \vec{\beta}_\infty$ for $t \to \infty$ by the continuity of the map $\vec{\beta}(\vec{E})$. 
        \item 
        If $\vec{e} \in \partial E_S$ and $\vec{E}_0 \in E_S$, again the entire flow $\vec{E}(t)$ lies within $E_S$ and (\ref{betasol}) is sensible for all $t \in [0,\infty)$. Thus,  $\Delta(\vec{E}(\vec{\beta}(t))) = \Delta_0 e^{-2 t}$ has a limit of $0$ for $t \to \infty$. Under the diffeomorphism $\vec{E}(\beta)$, the concentric spheres $|\vec{\beta}| = R$ for increasing $R$ map to concentric surfaces inside $E_S$, so for the flow $\vec{E}(t)$ which reaches the boundary of $E_S$ $\vec{\beta}(\vec{E})$ must eventually remain outside each surface $|\vec{\beta}| = R$. Thus $|\vec{\beta}(t)|$ diverges for $t \to \infty$.
        \item 
        If $\vec{e} \notin \partial E_S$ and $\vec{E}_0 \in E_S$, the linear segment between $\vec{e}$ and $\vec{E}_0$ must contain a point $p$ on the boudnary of $E_S$ in its interior. Thus, the trajectory (\ref{Eflow}) reaches the boundary of $E_S$ at some finite time $t_0$ and by the observations in the previous case, $|\vec{\beta}|$ must diverge here. In this case, $\vec{E}_p \ne \vec{e}$, so $\Delta(\vec{E_p}) \ne 0$ and $\lim_{t \to \infty} \Delta(\vec{E}(\vec{\beta}(t)) > 0$.  
    \end{enumerate}
\end{proof}
In practice, to find a state with specified expectation values, we can implement this flow numerically by discretizing the flow equation with some chosen time step. The quantities appearing on the right side of (\ref{floweq}) can be calculated explicitly by finding eigenvectors and eigenvalues $\{(|n \rangle, \lambda_n)\}$ for $M \equiv \vec{\beta} \cdot {\cal O}$ and calculating
\begin{eqnarray*}
    Z &=& \tr(e^{-M}) = \sum_n e^{-\lambda_n} \cr
    E_i &=& {1 \over Z} \tr({\cal O}_i e^{-M}) 
    = {1 \over Z} \sum_n e^{-\lambda_n} \langle n | {\cal O}_i | n \rangle \cr J_{ij} &=& -{1 \over Z}\int_0^1 ds \tr({\cal O}_i e^{-sM} {\cal O}_j e^{-(1-s)M}) + E_i E_j \cr
    &=& -{1 \over Z}\sum_{m,n} \langle m| {\cal O}_i | n \rangle \langle n| {\cal O}_j | m \rangle \int_0^1 ds \tr({\cal O}_i e^{-s \lambda_{m}} {\cal O}_j e^{-(1-s)\lambda_n }) + {1 \over Z^2} \tr({\cal O}_i e^{-M}) \tr({\cal O}_j e^{-M}) \cr
    &=& -{1 \over Z} \sum_{m,n} \langle m| {\cal O}_i | n \rangle \langle n| {\cal O}_j | m \rangle h(\lambda_m,\lambda_n) + E_i E_j
\end{eqnarray*}
where
\begin{equation}
    h(\lambda_m,\lambda_n) = \left\{ \ba{cc} {e^{-\lambda_m} - e^{- \lambda_n} \over \lambda_n - \lambda_m} & \lambda_m \ne \lambda_n \cr
    e^{-\lambda_n} & \lambda_m = \lambda_n \ea \right..
\end{equation}
We have have implemented the flow (\ref{floweq}) numerically using these expressions; Figure \ref{fig:flowplot}s provides an example shows the resulting flows starting from various $\vec{\beta}$.

There are various alternative flows we could consider with similar results. If $C(\vec{E})$ is any function on the space of expectation values with a unique minimum at $\vec{E} = \vec{e}$ and no other critical points, the function $C(\vec{E}(\vec{\beta}))$ will have a critical point if and only if $\vec{e} \in int(E_S)$, and the critical point will be a global minimum in this case. Thus, we could numerically implement the gradient flow for the function $C(\vec{E}(\vec{\beta}))$. The recent work \cite{enhanced_neg_energ} provides a detailed example where this method is employed to understand the space of allowed energy distributions in a quantum spin chain; there, minimization of the function $C(\vec{E}(\vec{\beta}))$ was carried out by a combination of gradient descent and Newton method algorithms. 

For $\vec{e} \in \text{int}(E_S)$, the method above provides a state $\rho_{\vec{\beta}_\infty}$ of the form (\ref{thermal}) with expectation values $\vec{e}$. Starting from this state, it is straightforward to describe the most general state with these expectation values. We define $(\mathfrak{h}_S)^\perp$ to be the subspace of $\mathfrak{h}_0$ orthogonal to $\mathfrak{h}_S = \text{span}(S)$. Then we have
\begin{proposition}
\label{generalstate}
  For a set $S$ of operators and a set of expectation values $\vec{e}$ in the interior of $E_S$, the most general quantum state that has these expectation values is
\begin{equation}
    \rho = {e^{- \vec{\beta}(\vec{e}) \cdot \vec{\cal O}} \over \tr(e^{- \vec{\beta}(\vec{e}) \cdot \vec{\cal O}})} + \rho^\perp
\end{equation}
where $\vec{\beta}(\vec{e})$ is the inverse of the map (\ref{Ebeta}) and $\rho^\perp$ lies in the subset of $(\mathfrak{h}_S)^\perp$ determined by requiring that $\rho$ is non-negative. 
\end{proposition}
\begin{proof}
For ${\cal O} \in S$, we have $\tr(\rho_1 {\cal O}) = \tr(\rho_2 {\cal O})$ so $\tr((\rho_1 - \rho_2) {\cal O}) = 0$. Thus, $(\rho_1 - \rho_2)$ is orthogonal to each operator in $S$ so is in $(\mathfrak{h}_S)^\perp$.
\end{proof}
Any specific direction away from the origin in $(\mathfrak{h}_S)^\perp$ corresponds to a one-parameter family of operators $\rho^\perp = \lambda {\cal O}^\perp$ for ${\cal O}^\perp \in (\mathfrak{h}_S)^\perp$. The allowed range of $\lambda$ for such a direction is the interval between the largest negative root and the smallest positive root of $p(\lambda) = \det(\rho_{\vec{\beta}(\vec{e})} + \lambda {\cal O}^\perp)$.

\subsubsection*{Application to the quantum marginal problem}

A particular example of Question 1 is the Quantum Marginal problem. For a  Hilbert space ${\cal H} = {\cal H}_A \otimes {\cal H}_B$
we can ask whether there is a state $\rho \in {\cal H}$ such that the reduced density operators $ \tr_B \rho = \rho_A$ and $\tr_A \rho = \rho_B$ for some specified $\rho_A$ and $\rho_B$. More generally, we can have a multipart Hilbert space and specify the reduced density operators for some or all of the proper subsystems. For the two-part case, the quantum marginal problem is equivalent to the question of whether $\vec{x} \in E_S$, where 
\begin{equation}
    S = \{\tilde{\cal T}^A_i \otimes \identity\} \cup \{\identity \otimes \tilde{\cal T}^B_i \} 
\end{equation}
and
\begin{equation}
    \vec{x} = (\tr(\rho_A \tilde{\cal T}^A_1), \dots, \tr(\rho_A \tilde{\cal T}^A_{N_A}),\tr(\rho_B \tilde{\cal T}^B_1), \dots, \tr(\rho_A \tilde{\cal T}^B_{N_B})  ) \; .
\end{equation}

\section{Uncertainty relations and positivity constraints}

Uncertainty relations in quantum mechanics, for example the Generalized Uncertainty principle (\ref{Uncert}), can be expressed as polynomial inequalities in the space of expectation values. All points in $E_S$ must satisfy these constraints. 

For the case of a two-dimensional Hilbert space, the uncertainty relation (\ref{Uncert}) applied to orthogonal operators $\hat{e}_1 \cdot \vec{\sigma}$ and $\hat{e}_2 \cdot \vec{\sigma}$ gives the relation
\begin{equation}
    \sum x_i^2  \le 1 +  (\hat{e}_1 \cdot \vec{x})^2 (\hat{e}_2 \cdot \vec{x})^2 \; .
\end{equation}
In this case, the set $E_S$ (described by $\sum x_i^2  \le 1$) is clearly contained in the set specified by this inequality, but we also have that the intersection of all such sets for different choices of $\hat{e}_i$ is exactly $E_S$, since the inequality implies in particular that $(\hat{e}_3 \cdot \vec{x})^2 \le 1$ for $\hat{e}_3 = \hat{e}_1 \times \hat{e}_2$. This motivates the question of whether $E_S$ in more general cases is determined by a collection of uncertainty relations.

To proceed, we recall that the uncertainty relation (\ref{Uncert}) can be understood as a special case of a more general set of relations that follow from the simple inequality
\begin{equation}
    \tr(\rho ({\cal O} - \bar{\cal O})^2) \ge 0 \qquad \qquad \bar{\cal O} = \tr(\rho {\cal O})
\end{equation}
which holds for any Hermitian operator ${\cal O}$. Taking ${\cal O} = v_i {\cal O}_i$ for some set of linearly independent Hermitian operators, we have that $\tr(\rho ({\cal O}_i - \bar{\cal O}_i)({\cal O}_j - \bar{\cal O}_j)) v_i v_j$ must be nonnegative for any $v_i$ so the matrix 
\begin{equation}
    \label{defMO}
    M_{ij} = \tr(\rho {\cal O}_i {\cal O}_j) - \tr(\rho {\cal O}_i) \tr(\rho {\cal O}_j)
\end{equation}
must be non-negative. Taking a pair of operators $\{{\cal O}_1,{\cal O}_2\}$, the non-negativity  of $M$ implies the non-negativity of the determinant, so we have
\begin{equation}
\left| \begin{array}{cc} \tr(\rho {\cal O}_1 {\cal O}_1) - \tr(\rho {\cal O}_1 ) \tr(\rho  {\cal O}_1) & \tr(\rho {\cal O}_1 {\cal O}_2) - \tr(\rho {\cal O}_1 ) \tr(\rho  {\cal O}_2) \cr \tr(\rho {\cal O}_2 {\cal O}_1) - \tr(\rho {\cal O}_2 ) \tr(\rho  {\cal O}_1) & \tr(\rho {\cal O}_2 {\cal O}_2) - \tr(\rho {\cal O}_2 ) \tr(\rho  {\cal O}_2) \end{array} \right| \ge 0.
\end{equation}
Expanding this out using ${\cal O}_1 {\cal O}_2 =  {1 \over 2} [{\cal O}_1, {\cal O}_2] + {1 \over 2} \{{\cal O}_1, {\cal O}_2\}$, we find 
\begin{align}\label{eq:unc_prin}
(\Delta {\cal O}_1)^2(\Delta {\cal O}_2)^2 \ge \left|\frac{1}{2}i\langle[{\cal O}_1,{\cal O}_2]\rangle\right|^2 +
    \left|\frac{1}{2}\langle\{{\cal O}_1,{\cal O}_2\}\rangle-\langle {\cal O}_1\rangle\langle {\cal O}_2\rangle\right|^2
\end{align}
where $\langle {\cal O} \rangle \equiv \tr(\rho {\cal O})$ and $(\Delta {\cal O})^2 \equiv \langle {\cal O}^2 \rangle - \langle {\cal O} \rangle^2$. The generalized uncertainty principle is an immediate consequence of this somewhat stronger bound. Thus, the non-negativity of $M_{ij}$ gives a generalization of the uncertainty relations.

We will now show that the non-negativity of $M_{ij}$ for a basis of traceless Hermitian operators ${\cal O}_i$ can be translated to a set of polynomial inequalities on $x_i = \tr(\rho {\cal O}_i)$ and that the complete set of these inequalities completely determines the region $E_S$ in this case. 

Taking $S= \{{\cal O}_i\}$ to be a basis of traceless Hermitian matrices, we have that any operator can be expressed as a linear combination of operators in $S \cup \{\identity\}$, so we can write
\begin{equation}
\label{defZ}
{\cal O}_i {\cal O}_j = Z^k_{ij} {\cal O}_k + {1 \over N}  g_{ij} \identity
\end{equation}
where $Z^k_{ij} = (Z_k^{ji})^*$ are complex coefficients. Then 
\begin{equation}
\label{defM}
    M_{ij} = Z^k_{ij} x_k + g_{ij} - x_i x_j \; . 
\end{equation}
For example, in a two dimensional Hilbert space with $S= \{\sigma_x, \sigma_y, \sigma_z\}$, we have
\begin{equation}
    M = \left( \begin{array}{ccc}
        1 - x^2 & iz - xy & -iy - xz \cr
         -iz - xy & 1 - y^2& -i - yz \cr
        iy - xz & -iz - yz & 1 - z^2
    \end{array} \right).
\end{equation}
Next, we can obtain a polynomial inequality by the requirement that any subdeterminant of $M$ is non-negative, or that $M_{ij} v_i v_j$ is non-negative for any $v$. In the two-dimensional example, the latter constraint (or the $1 \times 1$ subdeterminants) give $(\hat{e} \cdot x)^2 \le 1$ which together completely determine the Block sphere. Alternatively, non-negativity of any $2 \times 2$ subdeterminant gives $x^2 + y^2 + z^2 \le 1$.

We now show that the positivity of $M$ (equivalent to the set of polynomial relations that can be derived from this) completely determine $E_S$ for general dimentions.
\begin{theorem}
\label{uncert}
    Let $S= \{{\cal O}_i\}$ be a basis of traceless Hermitian operators and let $M_{ij}(\vec{x})$ be defined as in (\ref{defZ},\ref{defM}). Then $\vec{x} \in E_S$ if and only if $M_{ij}(\vec{x})$ is non-negative.
\end{theorem}
\begin{proof}
    If $\vec{x} \in E_S$, there is a state $\rho$ with $\tr(\rho {\cal O}_i)= \vec{x}_i$. For this state, $M_{ij} = \tr(\rho {\cal O}_i {\cal O}_j) - \tr(\rho {\cal O}_i) \tr(\rho {\cal O}_j)$ which we have already shown is non-negative. Now suppose $M_{ij}(\vec{x})$ is non-negative for some $\vec{x}$ and consider the operator $\rho_x = {1 \over N} \identity + \sum_i x_i \tilde{\cal O}_i$ where $\{\tilde{\cal O}_i\}$ is the dual basis of traceless Hermitian operators satisfying $\tr({\cal O}_i \tilde{\cal O}_j) = \delta_{ij}$. It is easy to check that $\tr(\rho_x {\cal O}_i) = x_i$, so $\vec{x} \in E_S$ provided that $\rho_x$ is a valid state. 
    The operator $\rho_x$ is Hermitian with unit trace by construction, so it remains to prove that $\rho_x$ is non-negative. Suppose that $v$ is a unit-normalized eigenvector of $\rho_x$ with eigenvalue $\lambda$. The operator ${\cal O}_v = v v^\dagger - \identity/N$ is Hermitian and traceless, so we can write it as ${\cal O}_v = \sum_i \alpha_i {\cal O}_i$ for some coefficients $\alpha_i$. Then by positivity of $M(x)$, we have
    \begin{eqnarray*}
        0 &\le & M_{ij}(\vec{x}) \alpha_i \alpha_j \cr
        &=& [\tr(\rho_x {\cal O}_i {\cal O}_j) - \tr(\rho_x {\cal O}_i) \tr(\rho_x {\cal O}_j)] \alpha_i \alpha_j \cr
        &=& \tr(\rho_x (v v^\dagger - {\identity \over N}) (v v^\dagger - {\identity \over N})) - \tr^2(\rho_x (v v^\dagger - {\identity \over N})) \cr
        &=& \lambda - \lambda^2.
    \end{eqnarray*}
    It follows that $\lambda$ is non-negative, so $\rho_x$ is a non-negative operator. 
\end{proof}
If $S$ is a smaller set of operators that is not a complete basis for the traceless Hermitian operators, it won't generally be true that the positivity of $M$ restricted to this subset (\ref{defMO}) completely determines $E_S$, but we could in principle make use of the positivity constraints for a basis $S_+$ of operators that includes $S$ to determine $E_{S_+}$ and then obtain $E_S$ by projection. In practice, we expect the characterizations of $E_S$ given in the previous sections will usually be more efficient.

\section*{Acknowledgements}
We would like to thank Jim Bryan and Brian Swingle for helpful discussions. We acknowledge support
from the National Science
and Engineering Research Council of Canada (NSERC) and the Simons foundation via a
Simons Investigator Award and the “It From Qubit” collaboration grant.

\appendix

\section{Algebraic characterization of $E_{\tilde{T}}$}

In this appendix, we derive explicit equations that determine $E_{\tilde{T}}$ as a subset of $\mathbb{R}^{N^2 -1}$.

To write the equations, it will be convenient to define coefficients $s_{abc}$ totally symmetric in $abc$ via
\begin{equation}
\label{defsabc}
   {1 \over 2} (\tilde{T}_a \tilde{T}_b + \tilde{T}_b \tilde{T}_a)  = N(\delta_{ab} \identity + s_{abc} \tilde{T}_c) \; .
\end{equation}

We can obtain equations characterizing  $\hat{E}_{\tilde{T}}$ in a few different ways.

\subsection{Equations from $\rho^2 = \rho$.}

Using the parameterization (\ref{defrho}) in the basic condition $\rho^2 = \rho$,  we find the operator equation
\begin{equation}
{1 \over N} (\identity + x_a \tilde{T}_a) {1 \over N} (\identity + x_b \tilde{T}_b) = {1 \over N} (\identity + x_c \tilde{T}_c) \; .
\end{equation}
Taking the components of this equation along the various $T$ basis directions gives our desired equations. The identity component (obtained by taking the trace of this equation) gives
\begin{equation}
\label{eq1}
\sum_a x_a^2 = N-1 \; .
\end{equation}
while the $\tilde{T}_a$ component (obtained via multiplying the equation by $\tilde{T}_a$ and taking the trace) gives
\begin{equation}
\label{eq2}
\left(1 - {2 \over N} \right) x_a = \sum_{b,c} s_{abc} x_b x_c \; .
\end{equation}
The equations (\ref{eq1}) and (\ref{eq2}) specify $\hat{E}_{\tilde{T}}$ though generally some of these are redundant. 

\subsection{Equations from the characteristic polynomial.}

For a projection operator with trace 1, the eigenvalues are $\{1,0,\dots,0\}$. So the characteristic polynomial is $\det(\lambda \identity - \rho ) = \lambda^{N-1} (\lambda - 1)$. Writing $\rho$ explicitly using the coordinate representation above and equating the coefficients of $\lambda^{N-1}, \dots \lambda^1, \lambda^0$ gives a set of $N$ equations that should characterize the surface $E_{pure}$.

More explicitly, defining $1/q = \lambda - {1 \over N}$ and using $\det(M) = \exp(\tr(\ln(M)))$, we have 
\begin{equation}
\exp(-\sum_n {1 \over n} \left({q \over N}\right)^n \tr((\sum_a x_a \tilde{T}_a)^n))  = (1 + {q \over N})^{N-1} (1 + {q \over N} - q).
\end{equation}
Equating the coefficients of $q^n$ on each side, we obtain a set of homogeneous equations 
\begin{equation}
    S_{a_1 \cdots a_M} x_{a_1} \cdots x_{a_M} = C_M \qquad \qquad M = 2, \cdots, N
\end{equation}
where 
\begin{equation}
    S_{a_1 \cdots a_M} x_{a_1} \cdots x_{a_M} = \sum_{\sum_n n k_n = M} (-1)^{\sum_n k_n} \prod_{n=2} {1 \over k_n! n^{k_n}} \left(\tr(x_a \tilde{T}_a^n) \right)^{k_n} 
\end{equation}
and 
\begin{equation}
    C_M = N(N-1)\dots(N-M+1) \; .
\end{equation}
The completely symmetric tensors $S_{a_1 \cdots a_n}$ can be expressed in terms of $s_{abc}$ and $\delta_{ab}$ using (\ref{defsabc}). 

For example, the equation for $n=2$ reproduces (\ref{eq1}), the equation for $n=3$ gives
\begin{eqnarray}
&&\tr((\sum_a x_a \tilde{T}_a)^3) = N(N-1)(N-2)\cr \implies &&N^2 s_{abc} x_a x_b x_c = N(N-1)(N-2)
\end{eqnarray}
and the equation for $n=4$ gives
\begin{eqnarray*}
&& \tr((\sum_a x_a \tilde{T}_a)^4) - \tr^2((\sum_a x_a \tilde{T}_a)^2) = N(N-1)(N-2)(N-3) \cr
\implies && (N^2(2N - 1) \delta_{ab} \delta_{cb} + 2 N^3 s_{abe} s_{cde}) x_a x_b x_c x_d  = N(N-1)(N-2)(N-3).
\end{eqnarray*}
Apparently, the collection of 2nd, 3rd, $\dots$, and $N$th order equations here are equivalent to the larger set of quadratic equations obtained from $\rho^2 = \rho$.

\subsection{Equations from vanishing $2 \times 2$ subdeterminants.}

Another equivalent set of equations are the equations that require the determinant of each $2 \times 2$ submatrix to vanish (required for rank 1). This gives for any $k > i$ and $l > j$ a quadratic equation
\begin{equation}
    (\identity + x_a \tilde{T}_a)_{ij} (\identity + x_a \tilde{T}_a)_{kl} - (\identity + x_a \tilde{T}_a)_{il} (\identity + x_a \tilde{T}_a)_{kj} = 0 \; ,
\end{equation}
where we have chosen some matrix representation for the operators $\tilde{T}$. As for the $\rho^2 = \rho$ equations, many of these are redundant.

\subsection{Equations in the canonical basis}

We can write a somewhat more explicit set of equations by choosing a basis $\hat{T}$ for $\mathfrak{h}$ such that we have the matrix representation
\begin{equation}
    \rho = \sum_a x_a \hat{T}_a =
    \left( \begin{array}{cccc} z_{11} & z_{12} & \dots & z_{1N} \cr z_{21} & z_{22} & \dots &z_{2 N}
    \cr
    \vdots & & \ddots & \cr
    z_{N1} & z_{N2} & \dots & z_{NN} \end{array} \right)
\end{equation}
where $z_{ii} = w_i$, $z_{ij} =  (u_{ij} - i v_{ij})/\sqrt{2}$ for $i<j$ and $z_{ij} =  (u_{ij} + i v_{ij})/\sqrt{2}$ for $i>j$. We are identifying $\{x_a\} = \{w_i,u_{ij},v_{ij}\}$. Here, we are no longer choosing the identity as one of the basis operators and we have the normalization
\begin{equation}
    \tr(\hat{T}_a \hat{T}_b) = \delta_{ab} \; .
\end{equation}

In this case, $\rho^2 = \rho$ together with $\tr(\rho) = 1$ give equations
\begin{equation}
    \sum_{j} z_{ij} z_{jk} = z_{ik} \qquad \sum_{i} z_{ii} = 1
\end{equation}
that characterize $\hat{E}_{\hat{T}}$ as a subset of $\mathbb{R}^{N^2}$. 

Alternatively, the vanishing of $2 \times 2$ subdeterminants together with $\tr(\rho) = 1$ give
\begin{equation}
    z_{ij} z_{kl} - z_{il} z_{kj} = 0 \qquad \sum_{i} z_{ii} = 1
\end{equation}
for $j<i$ and $k < l$. 

Finally, the condition $\det(\lambda \identity - \rho) = \lambda^{N-1} (\lambda - 1)$ give equations
\begin{equation}
    \sum_{i} z_{ii} = 1
\end{equation}
and 
\begin{equation}
    \sum_{\sum_n n k_n = M} (-1)^{\sum_n k_n} \prod_{n=1} {1 \over k_n! n^{k_n}} \left(\tr(\rho^n) \right)^{k_n} = 0
\end{equation}
for $M \ge 2$, where $\tr(\rho^n) = z_{i_1 \i_2} z_{i_2 i_3} \cdots z_{i_n i_1}$ with index summation implied.

\section{The Eigenset of $E_S$.}

\begin{figure}
    \centering
    \includegraphics[width = .9 \textwidth]{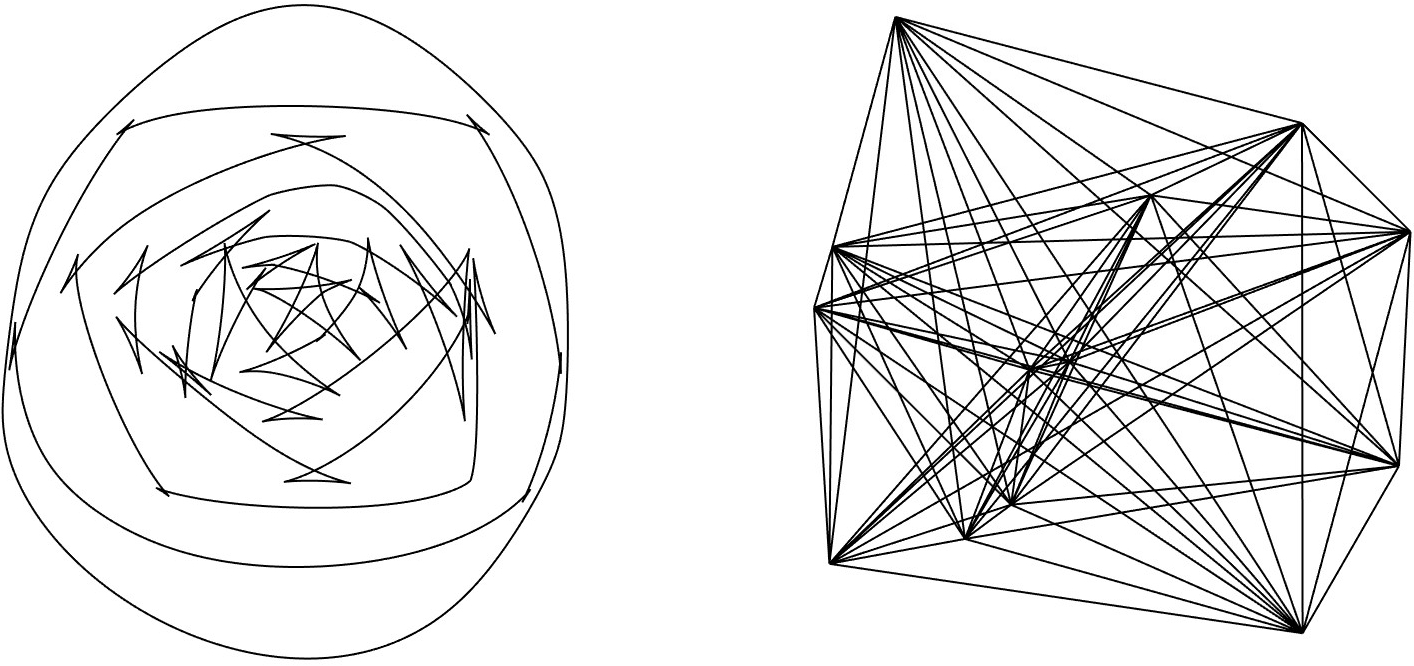}
    \caption{The set $V_S$ defined in (\ref{ESset}) for a typical pair of operators (left) and a pair of commuting operators (right) acting ${\cal H}$ with dimension 13. In each case, the outside edge of the figure is $\partial E_S$.}
    \label{fig:Vplot}
\end{figure}

We have defined $E_S$ as the set of possible expectation values for a set of $n$ Hermitian operators when considering all possible states. 

As an aside, it is interesting to define a subset of this, where we restrict to states that commute with some linear combination $\hat{e} \cdot \vec{O}$ of operators in $S$. These are states of the form $\rho = \sum_{ab} \rho_{ab} v_a v^\dagger_b$ where $v_a$ are eigenvectors of $\hat{e} \cdot \vec{O}$ with a common eigenvalue. Restricting this eigenvalue to the lowest eigenvalue of $\hat{e} \cdot \vec{O}$, we get exactly the boundary of $E_S$ according to the results of section 4. 
\begin{equation}
    \label{ESset}
    V_S = \bigcup_{\vec{e} \in \mathbb{R}^n \backslash 0} \{ \vec{x} | x_i = \tr(\rho {\cal O}_i), [\rho, \hat{e} \cdot \vec{O}] = 0 \} \; ,
\end{equation}  
This is generally the union of $\lfloor (n+1)/2 \rfloor$ closed surfaces, including $\partial E_S$ corresponding to eigenvectors with maximum/minimum eigenvalue and surfaces in the interior of $E_S$ corresponding to eigenvectors with the $k$th largest/smallest eigenvalues for some linear combination of operators in $S$. Examples of the set $V_S$ are shown in Figure \ref{fig:Vplot}.

\bibliographystyle{JHEP}
\bibliography{draft_bib}

\providecommand{\href}[2]{#2}\begingroup\raggedright\begin{thebibliography}{10}

\bibitem{wichmann1963density}
E.~H. Wichmann, \emph{Density matrices arising from incomplete measurements},
  {\emph{Journal of Mathematical Physics} {\bf 4} (1963) 884--896}.

\bibitem{Szymanski:2022sgn}
K.~Szyma\'nski, \emph{{Numerical ranges and geometry in quantum information:
  Entanglement, uncertainty relations, phase transitions, and state
  interconversion}}.
\newblock PhD thesis, Jagiellonian U., 2022.
\newblock \href{https://arxiv.org/abs/2303.07390}{{\tt 2303.07390}}.

\bibitem{enhanced_neg_energ1}
B.~Swingle and M.~Van~Raamsdonk, \emph{{Negative energies, wormholes, and
  cosmology}}, {\emph{QFARM seminar, June 2022, youtu.be/hisDmEB6q04} }.

\bibitem{enhanced_neg_energ}
B.~Swingle and M.~Van~Raamsdonk, \emph{{Enhanced negative energy with a
  massless Dirac field}},
  \href{http://dx.doi.org/10.1007/JHEP08(2023)183}{\emph{JHEP} {\bf 08} (2023)
  183}, [\href{https://arxiv.org/abs/2212.02609}{{\tt 2212.02609}}].

\bibitem{zeng2023maximum}
B.~Zeng, \emph{Maximum entropy methods for quantum state compatibility
  problems}, {\emph{Bulletin of the American Physical Society} (2023) },
  [\href{https://arxiv.org/abs/2207.11645}{{\tt 2207.11645}}].

\bibitem{Plaumann_2021}
D.~Plaumann, R.~Sinn and S.~Weis, \emph{Kippenhahn{\textquotesingle}s theorem
  for joint numerical ranges and quantum states},
  \href{http://dx.doi.org/10.1137/19m1286578}{\emph{{SIAM} Journal on Applied
  Algebra and Geometry} {\bf 5} (jan, 2021) 86--113}.

\bibitem{Szyma_ski_2018}
K.~Szyma{\'{n}}ski, S.~Weis and K.~{\.{Z}}yczkowski, \emph{Classification of
  joint numerical ranges of three hermitian matrices of size three},
  \href{http://dx.doi.org/10.1016/j.laa.2017.11.017}{\emph{Linear Algebra and
  its Applications} {\bf 545} (may, 2018) 148--173}.

\bibitem{eltschka2021shape}
C.~Eltschka, M.~Huber, S.~Morelli and J.~Siewert, \emph{The shape of
  higher-dimensional state space: Bloch-ball analog for a qutrit},
  {\emph{Quantum} {\bf 5} (2021) 485}.

\bibitem{bengtsson2012geometry}
I.~Bengtsson, S.~Weis and K.~{\.Z}yczkowski, \emph{Geometry of the set of mixed
  quantum states: An apophatic approach},  in \emph{Geometric Methods in
  Physics: XXX Workshop, Bia{\l}owie{\.z}a, Poland, June 26 to July 2, 2011},
  pp.~175--197, Springer, 2012.

\bibitem{goyal2016geometry}
S.~K. Goyal, B.~N. Simon, R.~Singh and S.~Simon, \emph{Geometry of the
  generalized bloch sphere for qutrits}, {\emph{Journal of Physics A:
  Mathematical and Theoretical} {\bf 49} (2016) 165203}.

\bibitem{sharma2021four}
G.~Sharma and S.~Ghosh, \emph{Four-dimensional bloch sphere representation of
  qutrits using heisenberg-weyl operators}, {\emph{arXiv preprint
  arXiv:2101.06408} (2021) }.

\bibitem{boya2008geometry}
L.~J. Boya and K.~Dixit, \emph{Geometry of density matrix states},
  {\emph{Physical Review A} {\bf 78} (2008) 042108}.

\bibitem{Swingle_2014}
B.~Swingle and I.~H. Kim, \emph{Reconstructing quantum states from local data},
  \href{http://dx.doi.org/10.1103/physrevlett.113.260501}{\emph{Physical Review
  Letters} {\bf 113} (dec, 2014) }.

\bibitem{yunger2016microcanonical}
N.~Yunger~Halpern, P.~Faist, J.~Oppenheim and A.~Winter, \emph{Microcanonical
  and resource-theoretic derivations of the thermal state of a quantum system
  with noncommuting charges}, {\emph{Nature communications} {\bf 7} (2016)
  1--7}.

\bibitem{Weis}
S.~Weis, \emph{Information topologies on non-commutative state spaces},
  {\emph{Journal of Convex Analysis} {\bf 21} (07, 2014) 339--399}.

\bibitem{rossmann2002lie}
W.~Rossmann and P.~Rossmann, \emph{Lie Groups: An Introduction Through Linear
  Groups}.
\newblock Oxford graduate texts in mathematics. Oxford University Press, 2002.

\end{thebibliography}\endgroup

\end{document}